\documentclass[journal]{IEEEtran}

\usepackage{amssymb,amsmath}
\usepackage{graphicx}
\usepackage{subfig}
\usepackage{algorithm}
\usepackage{algorithmic}
\usepackage{color, soul}
\usepackage{cite}
\usepackage{amsthm}
\usepackage{array}
\usepackage{cases}
\usepackage{url}

\DeclareMathOperator*{\argmin}{arg\,min}
\newtheorem{mydef}{Definition}
\newtheorem{mythm}{Theorem}
\newtheorem{mylem}{Lemma}

\newcommand{\eg}{\textit{e.g.}~}
\newcommand{\ie}{\textit{i.e.}~}
\newcommand{\etc}{\textit{etc.}~}
\newcommand{\cf}{\textit{cf.}~}

\newcommand{\one}{({\em i}\/)}
\newcommand{\two}{({\em ii}\/)}
\newcommand{\three}{({\em iii}\/)}
\newcommand{\four}{({\em iv}\/)}
\newcommand{\five}{({\em v}\/)}

\newcounter{LINumberOfComments}
\stepcounter{LINumberOfComments}

\newcounter{GTNumberOfComments}
\stepcounter{GTNumberOfComments}

\newcommand{\highlight}[1]{\textcolor{black}{#1}}
\newcommand{\gchange}[1]{\textcolor{black}{#1}}

\begin{document}

\title{FairCache: Introducing Fairness to ICN Caching\\
	Technical Report}

\author{Liang Wang$^{1}$ \quad Gareth Tyson$^{2}$ \quad Jussi Kangasharju$^{3}$ \quad Jon Crowcroft$^{1}$ \\
$^{1}$University of Cambridge, UK \quad $^{2}$Queen Mary University
London, UK \quad $^{3}$University of Helsinki, Finland}

\maketitle

\begin{abstract} 

Information-centric networking (ICN) is a popular research topic. At its heart is the concept of in-network caching. Various algorithms have been proposed for optimising ICN caching, many of which rely on collaborative principles, \ie multiple caches interacting to decide what to store. Past work has assumed altruistic nodes that will sacrifice their own performance for the global optimum. We argue that this assumption is flawed. We address this problem by modelling the in-network caching problem as a Nash bargaining game. We develop optimal and heuristic caching solutions that consider both performance and \emph{fairness}. We argue that only algorithms that are fair to all parties will encourage engagement and cooperation. Through extensive simulations, we show our heuristic solution, FairCache, ensures that all collaborative caches achieve performance gains without undermining the performance of others.


\end{abstract}



%
\IEEEpeerreviewmaketitle

\section{Introduction}
\label{sec:intro}


Information-Centric Networking (ICN)~\cite{jacobson:ccn,koponen:dona} has been proposed to exploit the observation that much of today's Internet traffic is content distribution.
ICN replaces the existing location-based Internet model with a content request/response model. One feature this enables is the capacity to cache content within the network. 
Whereas initial ICN caching approaches used traditional algorithms (\eg Least Recently Used), there has been a number of novel proposals that attempt to specifically target ICN environments. These algorithms exploit things like inter-AS cooperation, request prediction and a priori topology maps to optimise performance~\cite{Pacifici:2011:scr, 5062201, 6195634}.

A key outcome of this work has been the observation that collaborative caching usually outperforms locally optimised algorithms~\cite{Dahlin:1994:CCU:1267638.1267657, Chai:2012:CLM, 6566743, 6195634}. 
This is primarily caused by the nature of ubiquitous ICN caching, where nearby caches will often wastefully store the same objects~\cite{Chai:2012:CLM}. To address this, a simple collaborative algorithm might involve two nodes strategically caching distinct objects~\cite{wong:globecom2012, 5062201, Psaras:2012:PIC:2342488.2342501}. 
\gchange{
Cache collaboration is therefore likely to play a role in any future ICN deployments~\cite{wang:thesis:2015}. 
In tandem, we are witnessing a fragmentation of cache ownership in the Internet, with large operators deploying separate infrastructures. 
Some of these providers exclusively host their own content (\eg Google, Netflix), whilst other aggregate content from multiple sources (\eg Akamai, ChinaCache).
This adds an extra layer of complexity as its means that, even on an intra-domain level, we may begin to see multiple competing stakeholders operating caches within a single network. 
This will likely be accelerated by the growth of network function virtualisation, which will allow anybody to ``spin up'' caches within a network (we already see the availability of virtual cache services, \eg Fortinet Virtual Cache). Hence, ensuring the cache collaboration can also work in this setting will become increasingly important.}

A more extreme example of this fragmentation is within the expanding number of wireless community mesh networks, \eg Guifi~\cite{6379139}; these are deployed by groups of individuals who each contribute wireless routers (\eg mounted on their property). In a community network, \emph{every} router/cache is operated by a separate individual. Hence, we predict that future ICNs will use caches that are provisioned not just by network operators, but also various distinct stakeholders at strategic locations. These observations, however, have the potential to undermine the key tenets of caching in ICNs: What if caches operated by separate entities pursue policies that do not include collaboration, the storage of competitor's content or the serving of specific users? This is currently the situation online today, and it is unlikely to change with the advent of ICN. 
\gchange{Despite this, most ICN collaborative algorithm assumes altruistic nodes that simply strive to reach a global optimum~\cite{Psaras:2012:PIC:2342488.2342501, Chai:2012:CLM, 5062201, 6195634, borst:DistributedCaching_INFOCOM2010, 6566743, 6739053, Wang201548}. Whereas this is acceptabe in scenarios where a single organisation operates all caches, it ceases to be suitable in inter-domain scenarios or cases where caches in a single network are operated by multiple (third party) providers.}

The reasons why a non-collaborative policy may be implemented are diverse. However, in this paper, we explore the topic from a utilitarian perspective. Intuitively, caches would wish to engage in a collaborative algorithm if they attain greater utility than if they were not to engage. This observation mandates some concept of \emph{fairness}, where benefits are spread fairly across caches, and individuals are not expected to sacrifice personal performance by collaborating. Imagine, for instance, the above community network example; an individual who sees his/her own performance consistently detrimented by collaboration would (rationally) cease. We therefore argue that collaboration should be based on fairness, which may or may not reduce global performance. While a global optimum sounds attractive, we argue it is more important, from a practical perspective, that every node is better off by collaborating together than working alone.
In this paper, we design a collaborative caching algorithm that embraces both high performance and fairness. Our focus is not to build a protocol that forces nodes to collaborate, or provides protection against malicious behaviour but, rather, to design underlying algorithms that can fairly share cache space across trusted collaborators. We first formulate the fair in-network caching problem as a Nash bargaining game (\S\ref{sec:bargain}) before describing optimal algorithms for allocating objects to caches (\S\ref{sec:sol}). We then propose a heuristic collaborative caching algorithm (\S\ref{sec:fin}) with fairness at its core: FairCache. Through extensive simulations, we show that FairCache achieves in excess of 90\% accuracy compared to the optimal solution, at a fraction of the overheads (\S\ref{sec:exp}). Importantly, we show that, when using FairCache, \emph{all} nodes improve their performance via cooperation. It can be deployed across small subsets of collaborating caches or, alternatively, globally without change to design. We conclude by extracting key lessons learnt (\S\ref{sec:conclusion}).

\section{Motivation and System Model}
\label{sec:model}

To underpin our design, we begin with a motivational example before outlining our system model. For convenience, Table \ref{tab:notation} contains the notations used throughout the paper.


\subsection{Motivational Example}
\label{sec:example}


We use the simple toy caching system described in Figure~\ref{fig:example} as a motivating example. Imagine two routers with a cache capacity of one object. They each serve a nearby set of users and, consequently, it is desirable that they collaborate to decide which objects should be cached (\eg to avoid storing the same object). To decide which object to store, the caches locally inspect the request rates they receive, depicted in Figure~\ref{fig:example} (as a Demand Matrix). Intuitively, each cache would wish to selfishly optimise some concept of individual ``utility''. For simplicity, we measure their utility as the number of cache hits they get. We also allow nodes to redirect requests to the other cache; if a hit is attained there, a utility of 0.5 is given to the node performing the redirect (factored down due to the extra delay, overhead etc.). We consider three caching strategies:

\textit{Case 1: Greedy Strategy}, where each cache locally and selfishly optimises its performance. \highlight{As our comparison baseline, Greedy strategy is a perfect LFU that keeps track of all the objects.} Cache $1$ chooses to hold $A$ since it is the most popular content of demand $90$, which leads to $U_1 = 90 + \frac{31}{2} = 105.5$. Similarly, Cache $2$ caches $B$ which leads to $U_2 = 83 + \frac{5}{2} = 85.5$. Therefore, we have the aggregated utility $U_{Total} = U_1 + U_2 = 105.5 + 85.5 = 191$.

\textit{Case 2: Global Strategy}, where each cache tries to maximise the aggregated utility $U_{Total}$ of the whole system. By caching $C$ and $D$ on Cache $1$ and $2$ respectively, $U_{Total}$ reaches its theoretical maximum, namely $U_{Total} = 126 + 85 = 211$. However, if we examine the individual performance and compare them to the Greedy Strategy, we notice that the increase in utility for Cache $1$ results in a utility decrease for Cache $2$.

\textit{Case 3: Fair Strategy}, where caches attempt to collaborate fairly, in a way that does not reduce utility for any party. Cache $1$ stores $E$ and Cache $2$ stores $F$. Although this does not achieve the global optimum (\ie the aggregated utility $U_{Total}$ drops from $211$ to $208.5$), it ensures that \emph{both} caches improve their respective performance whilst also improving upon the local Greedy Strategy. This solution is Pareto efficient, and ensures both parties are incentivised. 

\begin{figure}[!tb]
  \centering
  \includegraphics[width=8.8cm]{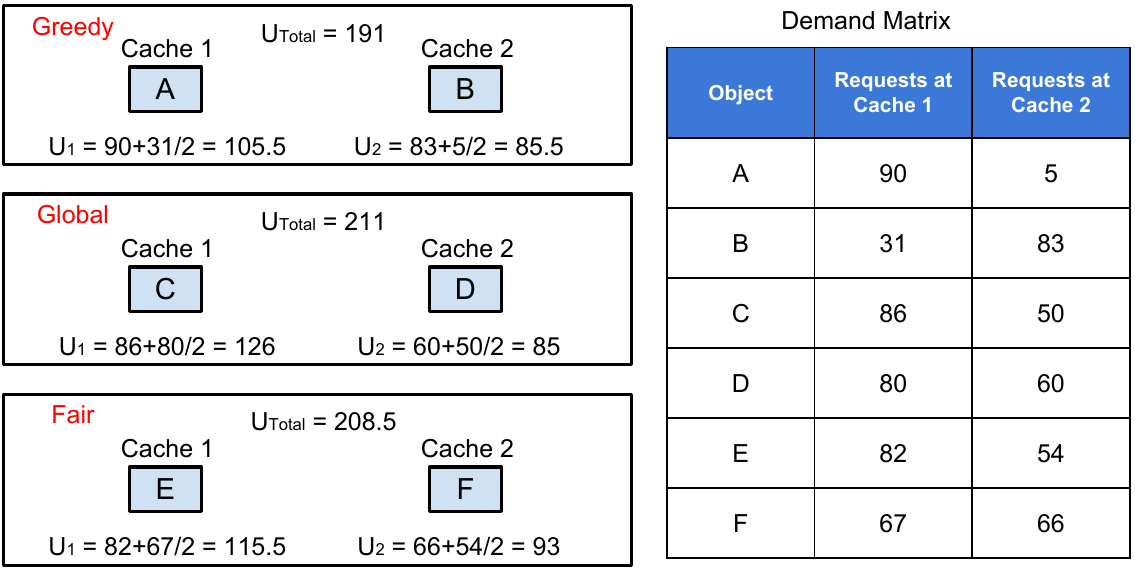}
    \caption{A mini caching system with two caches and six objects. Three strategies (Greedy, Global and Fair) are presented.}
  \label{fig:example}
\end{figure}

The above reveals a stark mismatch. Attaining a global optimum often disadvantages some parties~\cite{bertsimas2011price}. Thus, nodes that are unfairly exploited by other caches' redirects (at the cost of their own performance) are unlikely to continue collaboration: Caching should balance the need for high performance against the need for fair usage across caches. 



\begin{table}[!tb]
  \centering
  \begin{tabular}{  l  l  }
  \multicolumn{2}{c}{} \\
    \hline
    Notation                 	& Meaning \\
    \hline
    $G = (V, E)$           	& a graph G of node set $V$ and edge set $E$   \\
    $O$                           & content set, $o_k$ represents the $k^{th}$ item   \\
    $O'$        					& a reduced content set of $O$ by removing unpopular items  \\
    $\mathbf{s}$         		& $s_k$ is the size of the $k^{th}$ content item $o_k$  \\
    $\mathbf{w}$         	& demand matrix, $w_{i,k}$ is the demand of $o_k$ on $v_i$   \\
    $\mathbf{U}$         	& utility vector, $U_i$ is the utility of node $v_i$   \\
    $\mathbf{u^0}$    		& initial disagreement vector, $u^0_i$ is the disagreement value of $v_i$   \\
    $\mathbf{x}$          	& $x_{i,k}$ represents the decision whether $v_i$ caches $o_k$ locally  \\
    $\mathbf{y}$            	& $y_{i,j,k}$ represents the decision whether $v_i$ retrieves $o_k$ from $v_j$   \\
    $\mathbf{\Omega}$  	& solution space of all caching games,  $\mathbf{\Omega^e}$ is the Pareto frontier  \\
    $\mathbf{\Psi}$			& solution space of all fair caching games   \\
    $\boldsymbol{\lambda}$	& dual variable associated with the constraint (\ref{eq:fsb}): $y_{i,j,k} \leq x_{j,k}$   \\
    $\mathcal{L}(\cdot)$			&  Lagrangian associated with the objective function (\ref{eq:nash:max})   \\
    $d(\cdot)$     				& Lagrangian dual function of the objective function (\ref{eq:nash:max})    \\
    $h(\cdot)$ 					& subgradient of the dual function $d(\cdot)$  \\
    $\mathbf{\Phi}$			& overall communication complexity at system level   \\
    $N_i$                       	& neighbourhood of $v_i$, average size is denoted as $|\overline{N}|$   \\
    $N^+_i$        				& set of nodes having $v_i$ in their neighbourhoods   \\
    $r$                       		& average search radius, $r_i$ uniquely defines $N_i$ of $v_i$  \\
    $n_1$, $n_2$				& average size of one-hop and two-hop neighbourhood   \\
    $t$                      		& current iteration index while running subgradient descent   \\
    $\xi$                       	& $\xi_t$ is the step size of a subgradient algorithm at iteration $t$   \\
    $\theta'$, $\theta$ 	& current utility improvement and the stopping threshold   \\
    \hline
  \end{tabular}
  \caption{Table of main notation used in the paper}
  \label{tab:notation}
\end{table}

\subsection{System Model}

We model the network as a graph, $G = (V,E)$, where $V$ is the set of nodes and $E$ is the set of edges. $V$ could consist of all caches in a network or, alternatively, a subset of collaborating partners. These could be owned by one or more separate organisations. We follow an NDN~\cite{jacobson:ccn} model, whereby hosts generate requests that get deterministically routed to sources that reply with content objects. Each node in the network, $v_i \in V$, is equipped with cache of size $C_i$. We denote $O$ as the global set of content objects. For each $o_k \in O$, we associate two parameters: $s_k$, which is the object size and $w_{i,k}$, which is its aggregated demand (requests per second) observed from all the clients connected to $v_i$.
We focus on a subcategory of caching: collaborative algorithms. Because of resource constraints, we assume that nodes are limited in the number of nearby nodes they can cooperate with. We use $r_i$ to represent $v_i$'s search radius measured in hops. $r_i$ defines a neighbourhood for each $v_i$, which we denote as $N_i = \{v_j | l_{i,j}^* \leq r_i, \forall v_j \in V, v_i \neq v_j\}$, where $l_{i,j}^*$ measures the length of the shortest path between $v_i$ and $v_j$.

\gchange{A collaborative caching algorithm can be decomposed into ``caching decisions'' and ``retrieval decisions''. These two parts solve ``what to cache'' and ``where to fetch''. The latter is necessary to allow nodes to redirect requests to other caches (opposed to forwarding it along the default route to the original source). 
This means that caches that do not locally store an object retain the flexibility to exploit objects stored elsewhere (\ie collaboration).
To model such a caching strategy, we use two vector decision variables: $\mathbf{x}$ and $\mathbf{y}$. $x_{i,k} \in \{0,1\}$ denotes whether $v_i$ caches $o_k$, and $y_{i,j,k} \in \{0,1\}, \forall i \neq j$ denotes whether $v_i$ retrieves the object $o_k$ from $v_j$.} Formally, we say:


\begin{mydef}
  A caching strategy for a network $G$ is a tuple of functions
  $(\mathbf{x}, \mathbf{y})$ where $\mathbf{x}: V \times O \rightarrow
  \{0,1\}$ and $\mathbf{y}: V \times V \times O \rightarrow \{0,1\}$. The
  family of all such tuples is denoted as $\Psi$, which represents the
  whole space of all caching strategies.
\end{mydef}

\begin{mydef}
  A caching strategy for a node $v_i$ is defined as $(\mathbf{x_i},
  \mathbf{y_i})$, where $\mathbf{x_i}: \{v_i\} \times O \rightarrow
  \{0,1\}$ and $\mathbf{y_i}: \{v_i\} \times V \times O \rightarrow
  \{0,1\}$ are the partial functions of $\mathbf{x}$ and $\mathbf{y}$
  with domains restricted to $\{v_i\} \times O$ and $\{v_i\} \times V
  \times O$ respectively.
\end{mydef}

Note that $\times$ above represents the Cartesian product. We strive for a caching strategy that is \one~Pareto efficient; \two~has well-defined fairness achieved amongst the nodes; and \three~attains high performance. From a utilitarian perspective, this combination of attributes will lead to stable cooperation.




\section{Fundamentals of Bargaining Games}
\label{sec:bargain}

A bargaining game is a model for analysing how parties collaborate to obtain certain utility values. We model collaborative caching as a bargaining game, in which we aim to achieve both high performance (utility) and fairness. Ideally, a solution is considered fair if it satisfies certain axioms\cite{nash1950bargaining, RePEc:ecm:emetrp:v:43:y:1975:i:3:p:513-18}:  \one~Pareto optimality; \two~Scale invariance; \three~Symmetry; \four~Independence of the irrelevant alternatives; \five~Monotonicity. Nash proved that there is one unique solution which satisfies axiom \one-\four, termed the Nash Bargaining Solution (NBS)~\cite{nash1950bargaining}. The NBS can be extended to multiple players. On the other hand, the Kalai-Smorodinsky Solution (KS)~\cite{RePEc:ecm:emetrp:v:43:y:1975:i:3:p:513-18} satisfies axiom \one-\three and \five. These two solutions lead to two fairness metrics. Compared to NBS, KS often does not have a closed-form expression. Hence, we focus on NBS.


\subsection{Game Definitions}
\label{sec:bargain:nbs}

In game theory, each node attempts to optimise its personal ``utility''. In caching, utility for a node, $U_i$, can be measured by the delay to respond to a client. Each cache aims to serve its clients with the lowest possible delay. Consequently, serving a request from the local cache produces the greatest utility, but redirecting a request to another nearby neighbour also increases it (rather than forwarding to the original source). As such, a selfish cache strives to maximise its utility through a combination of local caching and redirects to nearby neighbours.\footnote{Here, we assume that each individual node selfishly optimises. However, our model can also support collective self interest amongst multiple nodes, \eg if several caches are owned by a single organisation.} Of course, if utility can be maximised solely through local caching then a node will cease to collaborate. NBS is an axiomatic solution for solving the following problem:
%
%
%
\begin{align}
  & \max \prod_{v_i \in V} (U_i - u^0_i) \label{eq:nash:1}
\end{align}
Eq.~(\ref{eq:nash:1}) is called the Nash product. As mentioned, $U_i$ is node $i$'s utility. $u_i^0$ is the initial disagreement value of $i$. The disagreement value is defined as the worst utility payoff a node would accept for collaboration. 
In practice, a node sets its disagreement value to the maximum value achieved by optimising locally as a standalone cache, \eg using Least Recent Used. In the following, we give the formal definition of our in-network caching game and its solution.

\begin{mydef}
An in-network caching game is a tuple $(\mathbf{\Omega},\mathbf{u^0})$, where $\mathbf{\Omega} \subset \mathbb{R}^{|V|}$ contains all the utility values obtainable via collaboration, and $\mathbf{u^0} \subset \mathbb{R}^{|V|}$ contains all the disagreement values leading to a negotiation breakdown. 
\end{mydef}

Let $\mathbf{\Omega^{e}} \subset \mathbf{\Omega}$ be the Pareto frontier of set $\mathbf{\Omega}$, \ie the potential Pareto efficient solutions. We assume that $\mathbf{\Omega^{e}}$ is also a concave function with a closed compact convex domain. A game is considered fair iff its outcome is fair. Therefore, we have:

\begin{mydef}
A fair caching game is a game $(\mathbf{\Omega}, \mathbf{u^0})$ with a Nash bargaining solution, \ie a function $f: \mathbf{\Omega^e} \rightarrow \mathbf{\Psi}$ such that $f(\mathbf{\Omega}, \mathbf{u^0}) = (\mathbf{x},\mathbf{y})$ uniquely maximises $\prod_{v_i \in V} (U_i - u^0_i)$.
\end{mydef}


By taking the logarithm of the objective function (\ref{eq:nash:1}), we have $\ln(\max \prod_{v_i \in V} (U_i - u_i^0)) = \max \ln(\prod_{v_i \in V} (U_i - u_i^0))$. 
By taking the negation, NBS can be obtained equivalently by:
\begin{align}
  & \max \sum_{v_i \in V} \ln (U_i - u_i^0) \Longrightarrow \min \sum_{v_i \in V} -\ln(U_i - u_i^0). \label{eq:nash:9}
\end{align}

\subsection{Fairness Definitions}
\label{sec:bargain:fairness}

We argue that collaboration should follow the intuitive concept of fairness, such that all caches receive fair utility improvements through collaboration. This is critical to ensure that node owners do not feel exploited and do not disengage from the collaboration. Being Pareto efficient, alone, does not achieve this. To attain fairness, it is necessary to formalise the concept. Three well-defined fairness metrics are often referred to in the literature~\cite{kelly1998rate, boche2009nash, muthoo1999bargaining}, \ie \textit{Egalitarian (EF)}, \textit{Max-min (MF)} and \textit{Proportional (PF)} fairness. \textit{EF} pursues an equal amount of improvement on every node, which usually creates Pareto inefficiency (and is thus seldom used in practice). Both \textit{MF} and \textit{PF} have axiomatic foundations and are widely used, \eg in traffic engineering. \textit{MF} is a generalisation of KS, while \textit{PF} is a generalisation of NBS. Thus, we only focus on \textit{PF}:

\begin{mydef}
\label{def:propfair}
Proportional Fairness (PF): A caching strategy $(\mathbf{x}^*,\mathbf{y}^*)$ is PF iff $\forall (\mathbf{x},\mathbf{y}) \neq (\mathbf{x}^*,\mathbf{y}^*) \Rightarrow \sum_{v_i \in V} \frac{U_{i} - U^*_i}{U^*_i - u^0_i} < 0$.
\end{mydef}

A cache allocation is considered \textit{PF} if the re-allocation of any object would decrease the proportional utility gain (from collaboration) of a node by \emph{less} than the respective aggregated increase for others. For example, imagine an object is re-allocated from Cache 1 to Cache 2. It would not be fair if this re-allocation reduces Cache 1's utility by 20\%, so that Cache 2 could increase its utility by just 3\%. However, it would be considered fair if Cache 2 could increase its utility by 60\%. Importantly, to be considered \textit{PF} (and to incentivise uptake), it is necessary for \emph{all} caches to improve their performance over local optimisation (\eg Least Recently Used). Otherwise collaboration would immediately breakdown in favour of local algorithms.
\gchange{If this were to occur, each node would simply select its own preferred algorithm (\eg LRU). However, in our caching games, \textit{PF} is guaranteed by NBS (the proof is trivial and available in~\cite{ourtechreport}).
It is also trivial to show whenever a Pareto efficient solution achieves EF, it also achieves MF, \ie a fair solution achieves all three fairness metrics in NBS given it is EF. We later show that being collaborative brings benefits to almost all nodes, indicating that the number of (rational) nodes who would revert to a local algorithm are very limited.}

\section{Solving a Fair In-Network Caching Game}
\label{sec:sol}

We next devise both centralised and decentralised optimal solutions for achieving fair caching. We later use these to evaluate our heuristic solution, \textit{FairCache}. 
We avoid presenting all the standard mathematical  details but rather focus on the key mechanisms (remaining are available in~\cite{ourtechreport}).

\subsection{Defining a Utility Function}

In this paper, we assume that a cache's utility is generated from serving its users' demand with low delay. For edge nodes, this demand comes directly from clients, whereas for core nodes this comes from their downstream customers. In either case, utility could be improved by a router using its local cache, or by redirecting a request to a nearby collaborative cache. Both improve delay compared to following the deterministic route to the origin. More precisely, $v_i$'s utility is defined as:
\begin{align}
  & U_i = \sum_{o_k \in O} s_k w_{i,k} x_{i,k} + \sum_{o_k \in O} \sum_{v_j \in N_i} \frac{s_k w_{i,k}}{l_{i,j}^* + 1} y_{i,j,k} \label{eq:util}
\end{align}
Both terms show that the utility is a non-decreasing function of demand and content size. The second term shows that the utility of retrieving remote content decreases as the distance increases. It indicates that a node prefers fetching from the closest source to reduce latency and traffic footprint. Although this affine utility function is used throughout the paper, any other metric (\eg bandwidth) or affine function can be used to model the utility without change to our model.


\subsection{Centralised Solution}
\label{sec:sol:central}

We begin by outlining the optimal solution, which can be computed centrally (\eg on a controller). Without loss of generality, we assume unit object size $s_k = 1$,\footnote{In practice, object size could either be varied per-object or, alternatively, objects can be separated into smaller fixed size units} also let $l_{i,j} \triangleq l_{i,j}^* + 1$ for simplicity of expression. Plugging in Eq.~(\ref{eq:util}) and Eq.~(\ref{eq:nash:9}), we define the optimisation problem as:
\begin{align}
	\min \sum_{v_i \in V} -\ln ( \sum_{o_k \in O} w_{i,k} x_{i,k} + \sum_{o_k \in O} \sum_{v_j \in N_i} \frac{w_{i,k}}{l_{i,j}} y_{i,j,k} - u_i^0 ). \label{eq:nash:max}
\end{align}
Subject to
\begin{align}
  & \sum_{o_k \in O} x_{i,k} \leq C_i, \quad \forall v_i \in V \label{eq:cache} \\
  & \sum_{v_j \in N_i} y_{i,j,k} \leq 1, \quad \forall v_i \in V, \forall o_k \in O \label{eq:dst} \\
  & y_{i,j,k} \leq x_{j,k}, \quad \forall v_i, v_j \in V, o_k \in O \label{eq:fsb} \\
  & x_{i,k} \in \{0,1\}, \quad \forall v_i \in V, o_k \in O \label{eq:int2} \\
  & y_{i,j,k} \in \{0,1\}, \quad \forall v_i, v_j \in V, o_k \in O \label{eq:int1}
\end{align}
Constraint (\ref{eq:cache}) means the content stored at a node cannot exceed its cache capacity. 
\highlight{Constraint (\ref{eq:dst}) simplifies the data scheduling and avoids requesting redundant content by constraining a node to retrieve a maximum of one complete object in a cache period.}
Constraint (\ref{eq:fsb}) says $v_i$ can retrieve $o_k$ from $v_j$ only if $v_j$ has cached it; it also says $v_i$ cannot get more than $v_j$ can offer. Constraints (\ref{eq:int2}) and (\ref{eq:int1}) impose the domain of decision variables. 

The above optimisation problem is a typical \textit{Integer Programming} program which is NP-Complete. By applying \textit{Linear Programming relaxation}, we relax constraints (\ref{eq:int2}) and (\ref{eq:int1}) by letting  $x_{i,k} \in [0,1]$ and $y_{i,j,k} \in [0,1]$.
We later round up/down $x_{i,k}$ and $y_{i,j,k}$ to construct caching strategies.
Such relaxation renders a suboptimal solution hence is considered as the lower bound of the actual performance.
Regarding Eq.~(\ref{eq:nash:max}), since all the affine functions are log-concave, their composite with logarithmic functions preserves concavity. Thus, the problem (\ref{eq:nash:max}) becomes a convex optimisation problem defined over a set of compact and convex constraints as Lemma \ref{thm:0} shows, which leads to the unique Pareto efficient solution, which is trivially followed by the existence of the equilibrium in NBS by definition\cite{nash1950bargaining}. 

\begin{mylem}
\label{thm:0}
  The problem (\ref{eq:nash:max}) is a convex optimisation problem.
\end{mylem}

The centralised solution can be derived by applying standard convex optimisation techniques (see \cite{ourtechreport}). The solver needs the demand matrix of each cache, cache size, content object set, whole network topology \etc as inputs. The whole equation system has $3|O| \cdot |V|^2 + 2|O| \cdot |V| + |V|$ variables and the same number of equations.

\begin{mythm}
\label{eq:thm:kkt}
  In a fair collaborative game, for the optimal caching strategy $(\mathbf{x^*_i}, \mathbf{y^*_i})$ of node $v_i$, there exist non-negative vectors $\boldsymbol{\alpha} \succeq 0$, $\boldsymbol{\beta} \succeq 0$, $\boldsymbol{\gamma} \succeq 0$, $\boldsymbol{\delta} \succeq 0$ and $\boldsymbol{\lambda} \succeq 0$, such that 
\begin{align}
  & x^*_{i,k} = \frac{1}{\alpha_i + \gamma_{i,k} - \sum_{v_j \in N^+_i} \lambda_{j,i,k}} - \frac{\tau_{i,k}}{w_{i,k}}  \label{eq:alc:1} \\
  & y^*_{i,j,k} = \frac{1}{\lambda_{i,j,k} + \beta_{i,k} - \delta_{i,k}} - \frac{l_{i,j} \tau_{i,k}'}{w_{i,j}}  \label{eq:alc:2}
\end{align}
where $\tau_{i,k} = U_i - u_i^0 - w_{i,k} x_{i,k}$ and $\tau_{i,k}' = U_i - u_i^0 - \frac{w_{i,k}}{l_{i,j}} y_{i,j,k}$.
\end{mythm}

The centralised solution can be easily derived by applying standard convex optimisation techniques, as its closed form expression shown in Theorem \ref{eq:thm:kkt}. The proof and actual equation system can be found in appendix. In practice, we can choose any modern LP solver to calculate the solution, but the computation of solving such a large system is non-trivial.

\subsection{Distributed Solution}
\label{sec:sol:distributed}

The optimal centralised solution has obvious drawbacks in its actual use: \one~it suffers from high computation complexity; \two~it creates a single point of failure; and \three~it is not adaptive under network dynamics. Hence, we next translate it into a distributed solution using decomposition techniques.

To solve an equation system, each node can be viewed as a subsystem. If they simply optimise locally, all the calculations in each subsystem are independent from those in others. However, due to collaboration, there are variables and constraints, which are referred to as \textit{complicating variables and constraints}~\cite{boyd2004convex}. These make calculations interdependent and couple a subsystem with others.
In problem (\ref{eq:nash:max}), the only complicating constraint is (\ref{eq:fsb}). 

To decompose Eq.~(\ref{eq:nash:max}), we apply Lagrangian dual relaxation. Lagrangian dual relaxation provides a non-trivial lower-bound of a primal. The difference between the dual and the primal is called the \textit{duality gap}, which can be zero if certain conditions are met as we show below. The Lagrangian $\mathcal{L}(\cdot): \mathbb{R}^{2|O||V|^2} \rightarrow \mathbb{R}$ associated with objective (\ref{eq:nash:max}) is defined as follows:
\begin{align}
  & \mathcal{L}(\mathbf{x},\mathbf{y},\boldsymbol{\lambda}) \label{eq:dual:1} \\
  \nonumber  & = \sum_{v_i \in V} [ - \ln(U_i - u_i^0) + \sum_{v_j \in N_i} \sum_{o_k \in O} \lambda_{i,j,k}(y_{i,j,k} - x_{j,k}) ] .
\end{align}
$\boldsymbol{\lambda} \succeq 0$ is the dual variable associated with constraint (\ref{eq:fsb}) of objective function (\ref{eq:nash:max}). Then the Lagrangian dual function $d(\cdot): \mathbb{R}^{|O||V|^2} \rightarrow \mathbb{R}$ is as follows:
\begin{align}
  & d(\boldsymbol{\lambda}) = \inf_{\mathbf{x} \in X, \mathbf{y} \in Y} \mathcal{L}(\mathbf{x},\mathbf{y},\boldsymbol{\lambda}). \label{eq:dual:2}
\end{align}
Given $\boldsymbol{\lambda}$, let $\mathbf{x^*}$ and $\mathbf{y^*}$ be the unique minimizers for the Lagrangian (\ref{eq:dual:1}) over all $\mathbf{x}$ and $\mathbf{y}$. Then the dual function (\ref{eq:dual:2}) can be rewritten as $d(\boldsymbol{\lambda}) = \mathcal{L}(\mathbf{x^*},\mathbf{y^*},\boldsymbol{\lambda})$. By maximising the dual function, we can reduce the duality gap. The Lagrangian dual problem of the primal (\ref{eq:nash:max}) is defined as:
\begin{align}
  & \max_{\boldsymbol{\lambda} \in \mathbb{R}^{|O||V|^2}} d(\boldsymbol{\lambda}) = \mathcal{L}(\mathbf{x^*},\mathbf{y^*},\boldsymbol{\lambda}). \label{eq:nash:5}
\end{align}
The constraints for the dual are the same as those of the primal except constraint (\ref{eq:fsb}) which has already been included in the dual objective function (\ref{eq:nash:5}).
Because (\ref{eq:nash:max}) is convex and all the constraints (\ref{eq:cache})(\ref{eq:dst})(\ref{eq:int2}) and (\ref{eq:int1}) are affine, Slater's condition holds given a solution exists, and the duality gap is zero. Thus, when the dual (\ref{eq:nash:5}) reaches its maximum, the primal also reaches its minimum. The optimal solution for primal problem (\ref{eq:nash:max}) can be derived from the optimal solution for dual problem (\ref{eq:nash:5}).

After decomposition, each node $v_i$ now only needs to optimise its utility locally for a given $\boldsymbol{\lambda}$ by calculating:
\begin{align*}
  &\min \mathcal{L}_i (\mathbf{x}, \mathbf{y}, \boldsymbol{\lambda}) \\
  &= - \ln(U_i - u_i^0) + \sum_{v_j \in N_i} \sum_{o_k \in O} \lambda_{i,j,k}(y_{i,j,k} - x_{j,k}).
\end{align*}
We use the standard projected subgradient method~\cite{boyd2004convex} to derive the algorithm. Let $h(\boldsymbol{\lambda})$ and $\partial d(\boldsymbol{\lambda})$ denote the subgradient and subdifferential of dual function $d(\cdot)$ at point $\boldsymbol{\lambda}$ respectively. Then for every $h_{i,j,k} \in h(\boldsymbol{\lambda})$ we have:
\begin{align*}
  & h_{i,j,k} = y_{i,j,k}^* - x_{j,k}^*  \Longrightarrow  h(\boldsymbol{\lambda}) \in \partial d(\boldsymbol{\lambda}).
\end{align*}
Gradient $\mathbf{h} \triangleq h(\boldsymbol{\lambda})$ points to the direction where $d(\cdot)$ increases fastest. In each iteration, node $v_i$ solves the local subsystem (\ref{eq:dual:3}) to update the dual variable $\boldsymbol{\lambda}$. $t$ represents the $t^{th}$ iteration. $\xi_t$ is the step-size in the $t^{th}$ iteration which can be determined by several standard methods~\cite{boyd2004convex}. The projected subgradient method projects $\boldsymbol{\lambda}$ on its constraint (\ie $\boldsymbol{\lambda} \succeq 0$) in each iteration, and we use $(\cdot)_+$ as a shorthand for the Euclidean projection. 
Eventually $\boldsymbol{\lambda}^{(t)} \rightarrow \boldsymbol{\lambda}^*$ when $t \rightarrow \infty$. The primal solution can be constructed from the optimum $\boldsymbol{\lambda}^*$.
Combining the above, we refer to Eq.~(\ref{eq:dual:3}) as the \textit{distributed optimal algorithm}:
\begin{align}
\label{eq:dual:3}
\begin{cases}
  & \mathbf{x}^{(t)}_i, \mathbf{y}^{(t)}_i = \argmin_{\mathbf{x},\mathbf{y}} \mathcal{L}_i (\mathbf{x}, \mathbf{y}, \boldsymbol{\lambda}^{(t)}) \\
  & \mathbf{h}^{(t)} = - (\mathbf{x}^{(t)}_i - \mathbf{y}^{(t)}_i) \\
  & \boldsymbol{\lambda}^{(t+1)} = (\boldsymbol{\lambda}^{(t)} + \xi_t \sum_{v_j \in N_i \cup \{v_i\}} \mathbf{h}_j)_+
\end{cases}
\end{align}


\begin{mythm}
\label{eq:thm:converge}
  Optimal algorithm converges to its optimum as the sequence $\{\boldsymbol{\lambda}^{(1)},\boldsymbol{\lambda}^{(2)}$ ... $\boldsymbol{\lambda}^{(t)}\}$ converges, if a diminishing step size is used such that $\lim_{t \rightarrow \infty} \xi_t = 0$ and $\sum_{t=1}^{\infty} \xi_t= \infty$.
\end{mythm}

The above theorem guarantees convergence\cite{ourtechreport}. $\lambda_{i,j,k}$ can be viewed as the \textit{``shadow price''} of transferring $o_k$ from $v_j$ to $v_i$, which is a ``cost'' for $v_i$ but an ``income'' for $v_j$.
\highlight{Given an ideal network, the distributed optimal algorithm will converge faster than the centralised one due to its parallel computations.}
It is worth emphasising that although the optimal algorithm above distributes the calculations over nodes, the overall computations are not reduced. At the same time, the communication cost increases due to exchanging ``shadow price'' information. However, the overall communication complexity remains the same as that of the centralised solution as we will later show in Section~\ref{sec:complexity}.

\section{FairCache: A Low-Complexity Heuristic Design}
\label{sec:fin}

The distributed algorithm, although optimal, comes with high overheads.
To mitigate this, we propose FairCache, a heuristic algorithm which does not require global knowledge regarding the content and network topology. 
We emphasise that FairCache is a decentralised algorithm for fairly sharing cache capacity across multiple \emph{trusted} stakeholders. It is not intended to be a protocol, by which malicious behaviour (\eg falsifying content demand) can be prevented. Hence, we assume trusted parties who faithfully execute the algorithm, much like is assumed within existing Internet routing schemes.


\subsection{Overview of Heuristics}
\label{sec:fin:rea}

To understand the rationality behind our heuristic, we first give a verbal explanation on the mechanisms of the optimal algorithm expressed in Eq.~(\ref{eq:dual:3}). Recall $\boldsymbol{\lambda}$ represents the shadow price of transmitting an object between two nodes. Each node hence maintains a list of prices for any given object from any given node. In each iteration of the optimisation, a node tries to minimise its total cost using $\boldsymbol{\lambda}^{(t)}$.
During the optimisation, the node adjusts its local caching strategy (via $\mathbf{x_i}$ and $\mathbf{y_i}$) and price list on other nodes (via $\boldsymbol{\lambda}$ and $\mathbf{h}$). Namely, a node may decide to cache an object if it brings significant improvement, or stop retrieving an object from another node due to high cost. 
Meanwhile the node adjusts the price on how to charge its neighbours by offering help. Then the node collects the price adjustments from \textit{all others in the network} and updates its own list. The procedure continues until the performance converges based on certain well-defined criteria (as described next). Future updates are periodically shared to address changes in content popularity. As, generally, popularity changes are relatively slow to occur (hours, rather than minutes), this does not create considerable overheads.
The mechanisms above indicate that we can approximate the optimal algorithm in the following ways:

\textbf{\one~Cut out unpopular content}: This approximation takes advantage of the highly skewed content popularity distribution. It is well-known that the popularity distribution has a long and heavy tail and most content fall into the tail. Removing the tail can significantly reduce the size of the exchanged messages. Meanwhile, the results will not be significantly influenced because of their marginal contribution to the overall utility (whilst also reducing signalling overheads dramatically). Thus, requests for unpopular content will be forwarded towards the origin (as with vanilla NDN~\cite{jacobson:ccn}).

\textbf{\two~Cut out distant nodes}:  This approximation takes advantage of topological locality. Since the utility of retrieving distant content is a decreasing function of the hop count between two nodes, the value quickly diminishes as path length increases. It is more likely to find the requested content in nearby nodes due to content spatial locality\cite{Wang:2015:PUS:2810156.2810162, Dabirmoghaddam:2014:UOC:2660129.2660143}; removing remote nodes should not have significant impact on the result.

\textbf{\three~Reverse direction}: This approximation takes advantage of the behaviours of gradient methods. In the optimal algorithm, the neighbourhood ($r$) gradually shrinks from the network diameter to its optimum (as a result of minimising the cost function). However, most elements in $\mathbf{y_i}$ are already set to zero by the gradient method in the beginning phase of the optimisation. Exchanging messages between nodes that are not going to collaborate is a waste of resources. By growing the neighbourhood set outwardly, instead of shrinking it, we can avoid unnecessary message exchange.



\subsection{FairCache Algorithm}
\label{sec:fin:algo}

\begin{algorithm}[!tb]
  \caption{Fair in-network caching (FairCache) on $v_i$}
  \label{alg:prjgrt}
  \begin{algorithmic}[1]
    \STATE{\textbf{Input:}}
    \STATE{\quad \quad Demand matrix $\mathbf{w}$}
    \STATE{\quad \quad Dual variables $\boldsymbol{\lambda}$}
    \STATE{\quad \quad Search radius $r = 0$}
    \STATE{\quad \quad Improvement threshold $\theta, \theta'$ $(\theta=10^{-2};\theta < \theta')$}
    \STATE{\textbf{Output:}}
    \STATE{\quad \quad Caching decision $\mathbf{x}_i$}
    \STATE{\quad \quad Collaboration decision $\mathbf{y}_i$}

	
	\WHILE{$\theta' \geq \theta$ \textbf{and} $r$ $<$ network diameter}
		\STATE{$r = r + 1$; $t = 0$};
    \WHILE{$t$ $<$ $t_{stop}$}
      \STATE{$\mathbf{x}_i$, $\mathbf{y}_i$ $= \argmin_{\mathbf{x},\mathbf{y}} \mathcal{L}_i(\mathbf{x}, \mathbf{y}, \boldsymbol{\lambda})$}
      \STATE{$\mathbf{h} = \mathbf{y}_i - \mathbf{x}_i$; trim $\mathbf{h}$ for $\forall v_j \in N^+_i$}
        \STATE{$\mathbf{h} = \mathbf{h} + \sum_{\forall v_j \in N_i} \mathbf{h}_j$}
      \STATE{$\boldsymbol{\lambda} = (\boldsymbol{\lambda} + \xi \mathbf{h})_+$}
      \STATE{$t = t + 1$}
    \ENDWHILE
    	\STATE{Update $\theta'$ with current improvement}
    \ENDWHILE
  \end{algorithmic}
\end{algorithm}

\begin{algorithm}[!tb]
  \caption{Construct reduced content set $O'_i$ on $v_i$}
  \label{alg:reduce}
  \begin{algorithmic}[1]
    \STATE{\textbf{Input:}}
    \STATE{\quad \quad Demand matrix $\mathbf{w}$}
    \STATE{\quad \quad Complete content set $O$}
    \STATE{\textbf{Output:}}
    \STATE{\quad \quad Reduced content set $O'_i$}
	
	\STATE{Sort $O$ in decreasing order based on $\mathbf{w}$;}
	\STATE{set $O'_i = \emptyset$;}
	\STATE{\textbf{for each}{ $o \in O$ }\textbf{do}}
		\STATE{\quad \textbf{if} {size ($O'_i$) + size ($o$) $\leq$ $v_i$'s cache capacity}}
			\STATE{\quad \textbf{then} add $o$ to $O'_i$}
		\STATE{\quad \textbf{else} break}
	\STATE{\textbf{end for each}}
  \end{algorithmic}
\end{algorithm}

\begin{algorithm}[!tb]
  \caption{On sending the price update $\mathbf{h}_i$ on $v_i$}
  \label{alg:send}
  \begin{algorithmic}[1]
    \STATE{\textbf{Input:}}
    \STATE{\quad \quad Neighbourhood $N^+_i$}
    \STATE{\quad \quad Caching strategy $\mathbf{x}_i$, $\mathbf{y}_i$}
	
	\STATE{set $\mathbf{h}_i = \mathbf{y}_i - \mathbf{x}_i$;}
	\STATE{\textbf{for each}{ $v_j \in N^+_i$ }\textbf{do}}
		\STATE{\quad construct $\mathbf{h}'_i$ from $\mathbf{h}_i $ based on $(O'_i \cup O'_j)$;}
		\STATE{\quad send $\mathbf{h}'_i$ to $v_j$;}
	\STATE{\textbf{end for each}}
  \end{algorithmic}
\end{algorithm}

\begin{algorithm}[!tb]
  \caption{On receiving all the $\mathbf{h}'_j$ on $v_i$}
  \label{alg:recv}
  \begin{algorithmic}[1]
    \STATE{\textbf{Input:}}
    \STATE{\quad \quad Neighbourhood $N_i$}
    \STATE{\quad \quad Caching strategy $\mathbf{x}_i$, $\mathbf{y}_i$}
	
	\STATE{set $\mathbf{h}_i = \mathbf{y}_i - \mathbf{x}_i$;}
	\STATE{\textbf{for each}{ $v_j \in N_i$ }\textbf{do}}
		\STATE{\quad receive $\mathbf{h}'_j$ from $v_j$;}
		\STATE{\quad set  $\mathbf{h}_i = \mathbf{h}_i + \mathbf{h}'_j$;}
	\STATE{\textbf{end for each}}
  \end{algorithmic}
\end{algorithm}

We embed the above heuristics in our algorithm, FairCache, presented in Algorithm~\ref{alg:prjgrt}. It takes several inputs.  $\mathbf{w}$ is the local demand matrix.
$r$ is for tracking the current number of hops that defines a node's neighbourhood radius.  $\theta'$ is used for recording the utility improvement by increasing the  radius from $r$ to $r+1$, while $\theta$ is the threshold below which FairCache should stop growing the neighbourhood size. $\boldsymbol{\lambda}$ is the list for tracking the shadow prices; this needs to be exchanged amongst nodes (via price adjustment $\mathbf{h}$) in a neighbourhood. Algorithm \ref{alg:reduce} shows how heuristic \one~is implemented to derive a reduced content set.

To apply the approximations, for node $v_i$, instead of making a complete price list, $\boldsymbol{\lambda}$, containing all the content and nodes in the network, node $v_i$ makes a partial $\boldsymbol{\lambda}$ which only includes: \one~the most popular content that can be fit into its local cache (\ie heuristic \one); and \two~the nodes in the neighbourhood defined by $r$ (\ie heuristic \two). It is possible that $v_i$ observes other content in the $\mathbf{h}_j$,  collected from neighbours while $r$ grows (\ie heuristic \three), then $v_i$ dynamically adds those content into its own $\boldsymbol{\lambda}$. $v_i$ can also remove items from $\boldsymbol{\lambda}$ if they are too expensive to retrieve. 
After local optimisation in each iteration, the price adjustment $\mathbf{h}$ will be trimmed before exchange by removing information that is not included in $\boldsymbol{\lambda}$; and removing the unchanged items, \ie the zero values. Essentially, $v_i$ only exchanges the trimmed $\mathbf{h}$ within its neighbourhood and $\boldsymbol{\lambda}$ only contains the aggregated popular content in the neighbourhood. Algorithm~\ref{alg:send} and \ref{alg:recv} detail the logic in lines 13--14 in Algorithm~\ref{alg:prjgrt} whenever sending and receiving price updates in each iteration.
Obviously, these approximations render incomplete information (due to removing unpopular content and distant nodes). To handle the missing $\lambda_{i,j,k}$ in the local optimisation, we let missing $\lambda_{i,j,k} = \infty$ ($i \neq j$), which indicates that the optimisation algorithm should neither exchange unpopular content nor exchange content with distant nodes.

Looking more closely, FairCache consists of two loops. The outer loop (lines 9--19) increases the search radius $r$ by one hop in each iteration. The outer loop stops when the current improvement, $\theta'$, drops below the threshold $\theta$ (\ie $\theta' < \theta$) due to enlarging the neighbourhood. The inner loop (lines 11--17) finishes the local optimisation ---- this is the calculation in Eq.~(\ref{eq:dual:3}) for the given neighbourhood defined by the current radius $r$. The communication overhead come from the operation in line 14 which collects the price adjustments $\mathbf{h}_j$ from the neighbourhood $N_i$. Line 15 adjusts the local shadow price list and updates the $\boldsymbol{\lambda}$ by removing or adding items.

\section{FairCache Complexity Analysis}
\label{sec:complexity}

\subsection{Overview}


The distributed solution requires exchange of control messages between nodes, whilst the centralised solution requires the distribution of the derived solution to \emph{all} nodes (once it has been centrally computed). The optimal solution comprises of $|O|$ matrices of size $|V|^2$. The aggregation of each row $i$ (in all matrices) represents one specific optimal strategy for the corresponding node $v_i$ in the system. To transmit these rows to the nodes, the aggregated communication complexity is $\Theta(|O| \cdot |V|^2)$. This is the same for the distributed solution, as this performs the same computation (shared across multiple nodes). The rest of this section will delineate the complexity of FairCache. Table~\ref{tab:complexity} presents an overview of the computational, communication, and space complexity of all three solutions.

\subsection{FairCache Heuristic Complexity}

\highlight{FairCache's computation, communication, and space complexities all correlate with the size of the content set and the neighbourhood set.} Herein we evaluate the complexity savings of FairCache's three heuristics.

Heuristic \one~is used to reduce the content set size ($|O'|$ from $|O|$). The level of reduction depends on both the popularity distribution of content and the percentage of requests we want to capture in the request streams. It is well documented that content popularity follows a Zipf-like distribution~\cite{Cha:2007:ITY:1298306.1298309, 4539688}. If we assume a Zipf distribution with $\alpha = 1.0$, for a content set of $10^6$ \highlight{equal sized objects}, we are able to cover $72.8\%$ of the requests by caching only $2\%$ of $|O|$. In other words, $|O'| = 0.02 \cdot |O|$ indicates a $98\%$ reduction. 
\highlight{Even if we only cache $1000$ objects (\ie $0.1\%$ of $|O|$), we can still capture $52\%$ of requests, whilst leading to a significant reduction in $|O'|$, \ie up to $99.9\%$.}


Heuristic \two~is used to reduce the size of collaboration neighbourhood, which is $|V|$ for the distributed solution.
By investigating the FairCache algorithm, we can see that the communication complexity is due to exchanging  $\mathbf{h}$ in order to update the local shadow price $\boldsymbol{\lambda}$, \highlight{\cf the last equation in Eq.(\ref{eq:dual:3}) and its corresponding line 14 in Algorithm \ref{alg:prjgrt}. The overhead therefore consist of two parts.}
The first part is induced by replying the queries on $\mathbf{h}$ from the nodes having $v_i$ in their neighbourhood, namely $N^+_i$.
The second part is induced by collecting $\mathbf{h}$ from the nodes in $v_i$'s own neighbourhood, namely $N_i$. Given that the communication complexity is measured by the number of exchanged messages, the complexity $\phi_i$ of node $v_i$ can be calculated as
\begin{align}
  & \phi_i = c \cdot |O'| \cdot (|N^+_i| + |N_i|) \label{eq:ovh:1}
\end{align}
Scalar $c$ in Eq.~(\ref{eq:ovh:1}) represents a constant factor for communication complexity, and can be understood as message size or other protocol-dependent factors. Now we show how to calculate the overall complexity $\Phi$ at system level from $\phi_i$ in the following.


Given any $v_j \in N_i$, a neighbourhood relation can be written as a tuple $(v_i, v_j)$. Calculating $\sum_{v_i \in V} |N_i|$ is equivalent to counting how many tuples there are in the whole system. Obviously, $v_i \in N_j \Longleftrightarrow v_j \in N^+_i, \forall v_i, v_j \in V$, in other words, as long as there is a tuple $(v_i,v_j)$ for $N_i$, there must be a tuple $(v_j,v_i)$ for some $N^+_j$, and vice versa. Hence we can show $\sum_{v_i \in V} |N_i| = \sum_{v_i \in V} |N^+_i|$. Because one tuple $(v_i, v_j)$ represents one message exchange from $v_i$ to $v_j$, by using double counting technique, it is obvious that each message will be counted twice in the calculation if we try to aggregate all $\phi_i$: \ie $(v_i,v_j)$ is in both $v_i$'s $ N_i$ and $v_j$'s $ N^+_j$. Therefore, overall complexity $\Phi$ can be calculated as half of the aggregated $\phi_i$ of all $v_i$ in $V$ as below.
\begin{align}
  \Phi = \frac{1}{2} \sum_{v_i \in V} \phi_i &= \frac{c}{2} \cdot |O'| \cdot (\sum_{v_i \in V} |N^+_i| + \sum_{v_i \in V} |N_i|) \\ &= c \cdot |O'| \cdot \sum_{v_i \in V} |N_i| \label{eq:ovh:3}
\end{align}

The intuitive explanation of eliminating $N^+_i$ in calculating $\Phi$ is that we only need to count the aggregated messages sent out from a node (\ie the size of a node's own neighbourhood) since for each message there is always a correspondence in the network to receive it.
Clearly, if we denote the average neighbourhood size as $|\overline{N}|$, then the system level communication complexity of FairCache is $\Phi = \Theta(|O'| \cdot |V| \cdot |\overline{N}| )$ as below.
\begin{align}
  \Phi = c \cdot |O'| \cdot \sum_{v_i \in V} |N_i| &= c \cdot |O'| \cdot \sum_{v_i \in V} |\overline{N}| \\
  &= \Theta(|O'| \cdot |V| \cdot |\overline{N}| ) \label{eq:ovh:4}
\end{align}

Similarly in the distributed solution Eq.(\ref{eq:dual:3}), since each node has the whole network of size $|V|$ as its neighbourhood and uses the complete content set $O$ in the optimisation. Then its complexity can be calculated as $\Phi = \Theta(|O| \cdot |V|^2)$ by plugging the term $|\overline{N}| = |V|$ and $|O'| = |O|$ into Eq.(\ref{eq:ovh:4}). As we can see, the derived complexity is exactly the same as that of the centralised solution.


Heuristic \three~is used to minimise the size of the neighbourhood by increasing, rather than decreasing, the neighbourhood size. As we can see from the analysis above, the complexity of FairCache is proportional to its average neighbourhood size $|\overline{N}|$. Reducing its size is therefore beneficial. The authors in \cite{Wang:2015:PUS:2810156.2810162} have proposed a neighbourhood model to calculate $|\overline{N}|$ in its closed form on any general network topologies. With its most general form, the overall complexity of the network of average $r$-hop neighbourhood can be calculated as 
\begin{align*}
  \Phi &= |O| \cdot |V| \cdot n_1 \cdot \Big( 1 + \left[ \frac{n_2}{n_1} \right] + \left[ \frac{n_2}{n_1} \right]^{2} + \cdots + \left[ \frac{n_2}{n_1} \right]^{r-1} \Big) \\
   &= \Theta \Big( |O| \cdot |V| \cdot \left[ \frac{n_2}{n_1} \right]^{r} \Big) \qquad \qquad \forall n_1, n_2 \in \mathbb{N}, n_2 > n_1
\end{align*}
where $n_1$ and $n_2$ denote the average number of one-hop and two-hop neighbourhoods, and $r$ represents the search radius.  If we assume $n_2 > n_1$ (which holds for all scale-free networks), the complexity of FairCache will grow exponentially if the search radius $r$ increases.\footnote{Note that most natural graphs like Internet, ISP networks, and social networks are all scale-free~\cite{SpringN:Rocketfuel}} Hence, by ensuring a small $r$, Heuristic \three~dramatically improves scalability. Importantly, it has been shown that $r$ is very small in most network optimisation problems~\cite{Wang:2015:PUS:2810156.2810162} (confirmed in Section \ref{sec:exp}).

\highlight{Besides the memory for storing the actual content objects, extra space is needed to store the input parameters of the optimisation problem (\ie $\mathbf{w}$ and $\mathbf{\lambda}$). The space complexity specifically refers to such extra space. Obviously, the space complexity is also decided by the size of content set and the size of neighbourhood. Hence FairCache on each node needs $\Theta(|O'| |\overline{N}|)$ memory space to keep track of the information needed for optimisation. If we assume the average object size is $800$~KB (much smaller than the average YouTube video size: $8$~MB \cite{Cha:2007:ITY:1298306.1298309}) and each node has $8$~GB memory installed. Each node needs to maintain roughly $10^4$ objects. If we allocate $64$-bit space to store the information for each object, given a neighbourhood of $10$ nodes, even in the worst case wherein all these content sets are disjoint, each node only needs approximately $800$~KB extra memory. However, spatial locality is very common \cite{Wang:2015:PUS:2810156.2810162} in content networks, which means there is a significant overlap regarding the stored content among nearby routers. If we assume that only $10\%$ of objects are overlapping, the $800$~KB can be further reduced to $651$~KB (by summing up a geometric sequence to calculate). In the end, the extra memory overhead (\ie the ratio between the memory space for parameters and memory space for objects) is less than $0.01\%$.}



\subsection{Complexity Summary}

Table~\ref{tab:complexity} summarises our previous analysis of the complexity of the three different solutions. We wish to highlight the following key points:

\begin{itemize}

\item \highlight{The computation and space complexity of the centralised and distributed solutions are identical.} It is always $\Phi = \Theta(|O| \cdot |V|^2)$. However, because the distributed solution is able to share the complexity across each node in the network, in practice, it leads to a more scalable design and avoids the single point of failure.

\item FairCache is able to reduce all complexities by limiting the sizes of the content set and its collaborative neighbourhood. The complexity of FairCache is only a small fraction of the distributed solution, more precisely $\frac{|O'|}{|O|} \cdot \frac{|\overline{N}|}{|V|}$. We have shown $|O| \gg |O'|$, also in reality $|V| \gg |\overline{N}|$ (as showed in \cite{Wang:2015:PUS:2810156.2810162}). We also confirm this experimentally in Section \ref{sec:exp}), hence the improvement is significant and improvements should grow as the network size increases.

\item \highlight{FairCache on each node needs $\Theta(|O'| |\overline{N}|)$ memory space to keep track of the information needed for optimisation (\ie $\mathbf{w}$ and $\mathbf{\lambda}$). In practice this introduces only negligible overhead. As the average object size is bigger, the overhead becomes even smaller.}

\end{itemize}

\begin{table}[!tb]
  \centering
  \begin{tabular}{  l  l  l  l  }
  \multicolumn{4}{c}{Complexity comparison of three different solutions} \\
    \hline
    Type of complexity & Centralised & Distributed & FairCache  \\
    \hline
    Computation        & $\Theta(|O| |V|^2)$  & $\Theta(|O| |V|^2)$  & $\Theta(|O'| |V| |\overline{N}|)$ \\
    Communication    & $\Theta(|O| |V|^2)$  & $\Theta(|O| |V|^2)$  & $\Theta(|O'| |V| |\overline{N}|)$ \\
    Memory space     & $\Theta(|O| |V|^2)$  & $\Theta(|O| |V|^2)$  & $\Theta(|O'| |V| |\overline{N}|)$ \\
    Individual (worst) & $\Theta(|O| |V|^2)$  & $\Theta(|O| |V|)$  & $\Theta(|O'| |\overline{N}|)$ \\
    \hline
  \end{tabular}
  \caption{The overall and (worst) individual complexity of the three solutions. The communication complexity of the centralised solution does not occur during the computation but comes afterwards when distributing the solution to all nodes. \highlight{The distributed solution does not reduce any complexity of the centralised solution but, instead, shares it across multiple nodes. FairCache significantly reduces all three complexities by applying its heuristics.}}
  \label{tab:complexity}
\end{table}

\section{FairCache Evaluation}   
\label{sec:exp}

\subsection{Methodology}
\label{sec:method}

To evaluate FairCache, we perform extensive simulations using the publicly available LiteLab platform~\cite{7444732}. We use several topologies. First, we use real topologies collected by the Rocketfuel project~\cite{SpringN:Rocketfuel}; namely, two ISP router-level topologies: Sprint (604 nodes, 2,279 edges) and AT\&T (631 nodes, 2,078 edges). \gchange{This embodies the usecase where an individual network has allowed the deployment of caches in their network by multiple parties (\eg via network function virtualisation).}
Second, we use traces from Guifi~\cite{6379139}. Guifi is the largest open wireless community mesh network in the world. It allows any user to purchase equipment and become part of the network. We use its core network topology in the  Catalunya region (735 nodes, 1,059 edges). 
\highlight{Third, we use the CAIDA ITDK trace\cite{caidatopo2012} which takes a snapshots of the Internet AS topology every 24 hours. We use the 7/01/2017 trace (25,107 nodes, 49,458 edges).
This captures the situation where each AS operates its own caches, but chooses to collaborate with neighbouring ASes.}
Fourth, to allow us to vary key graph parameters, we also generate synthetic networks based on two models: the Barab\'{a}si-Albert (BA) model and the Erd\H{o}s-R\'{e}nyi (ER) model. Four parameter sets are used for these synthetic networks: $\{BA_1: m=2\}$, $\{BA_2: m=4\}$, $\{ER_1: p=1.1\cdot \log(n)/n\}$ and $\{ER_2: p=1.5 \cdot \log(n)/n\}$.
\highlight{If there are multiple components in a synthetic network, we only use the largest one. $\log(n)/n$ in the setting guarantees there is only one single giant component in the network with high probability\cite{bornholdt2003handbook}. We have also analysed many recent AS-level topologies using the CAIDA ITDK traces~\cite{caidatopo2012}. Since most AS networks have an average degree of $4$, we use $BA_2$ configuration to reflect this fact.} \gchange{Therefore, we note that our synthetic networks represent both a per-router and per-AS topologies.}
For each topology we attach a single client to each edge router (\ie with degree of 1). For example, this results in 161 clients in Sprint; 207 clients in AT\&T; and 200 in Guifi. We then randomly select between 10 and 20 distinct routers to attach a source to. Each router is then allocated a given cache capacity, which we vary; the default is 4 GB, which is $<0.1\%$ of the corpus. We select a low value to be representative of feasible cache capacity in a wide area network with a large corpus.


Using the above topologies, clients generate requests at each simulation tick, which are then routed through the network to either a content source or an intermediate cache based on the strategy employed. We base our content set on the Youtube trace from~\cite{Cha:2007:ITY:1298306.1298309}. This contains 1,687,506 objects (average size is $8.0$~MB and aggregated size is $12.87$~TB). We use the view count information to fit a \textit{Zipf}$(\alpha)$ distribution $(\alpha = 0.9537)$ to model the overall content popularity. To explore the impact of different request patterns, we also perform sensitivity analysis on $\alpha$. 
Throughout this section, we use our distributed optimal algorithm (\ie Eq.~(\ref{eq:dual:3})) as an optimal benchmark to compare FairCache against. Each result is averaged over 50 runs; errorbars are not plotted if they are sufficiently small ($<5\%$).


\subsection{Scalability}
\label{sec:exp:con}

\begin{figure*}[!htp]
  \centering 
  \subfloat[Convergence on ISP networks.]{\label{fig:convergence:3}\includegraphics[width=4.5cm]{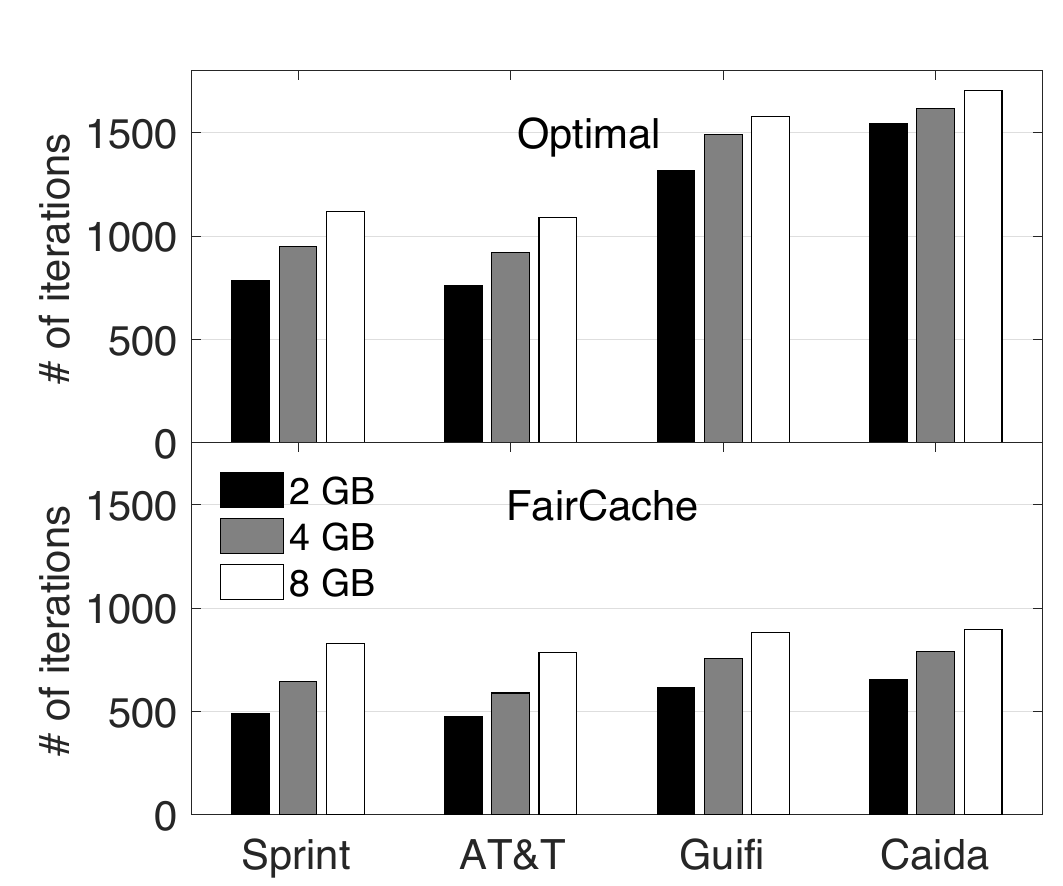}}
  \subfloat[Convergence on synthetic networks.]{\label{fig:convergence:2}\includegraphics[width=4.5cm]{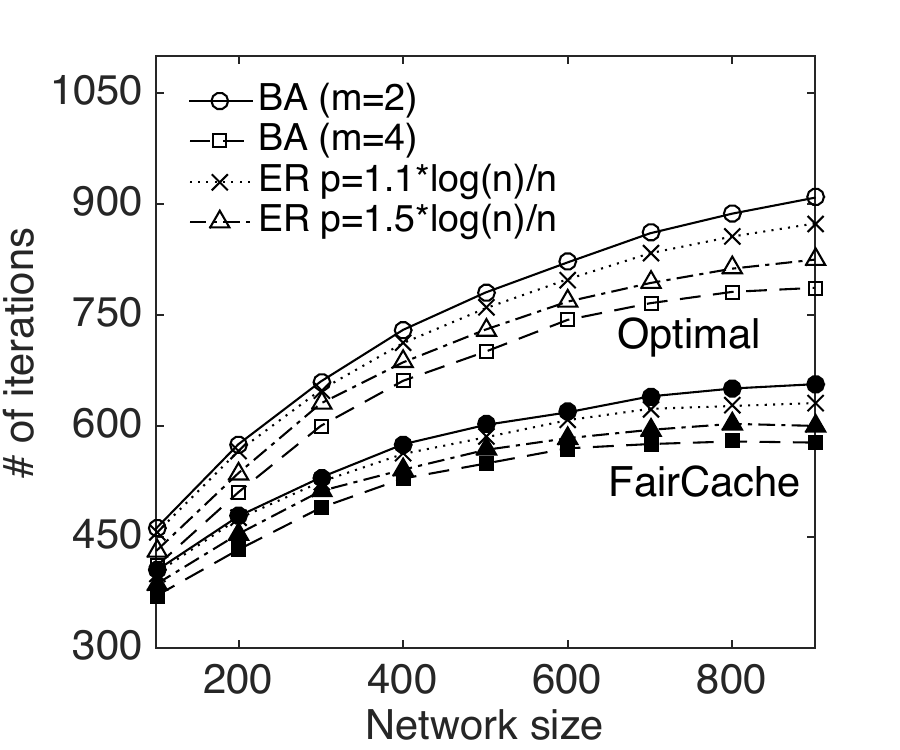}}
    \subfloat[Reduction of traffic overhead.]{\label{fig:cost:1}\includegraphics[width=4.5cm]{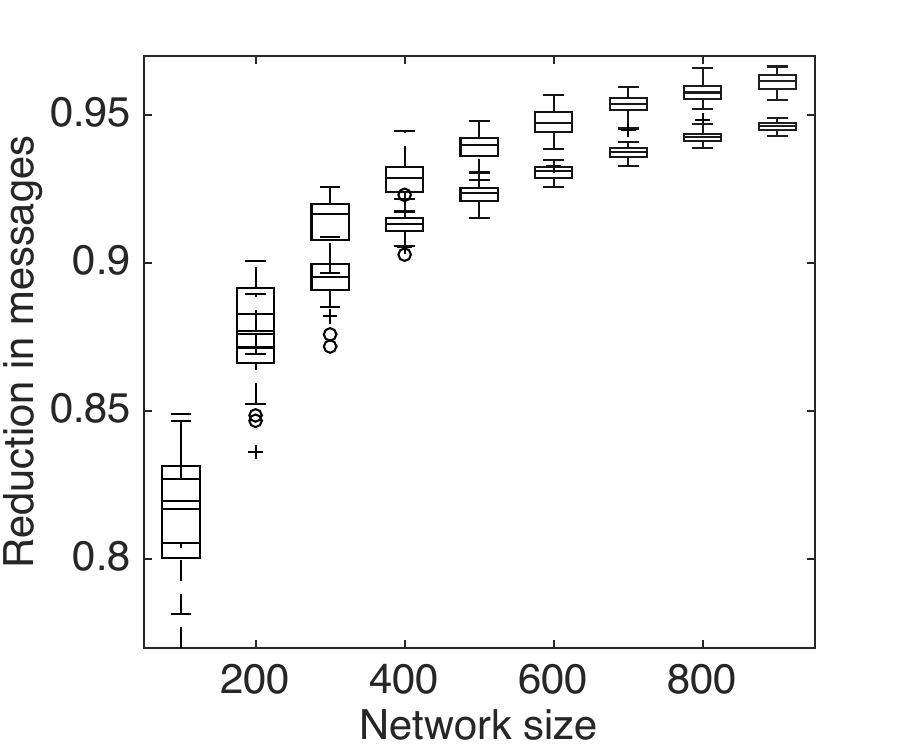}}
    \subfloat[The accuracy of FairCache.]{\label{fig:accuracy}\includegraphics[width=4.5cm]{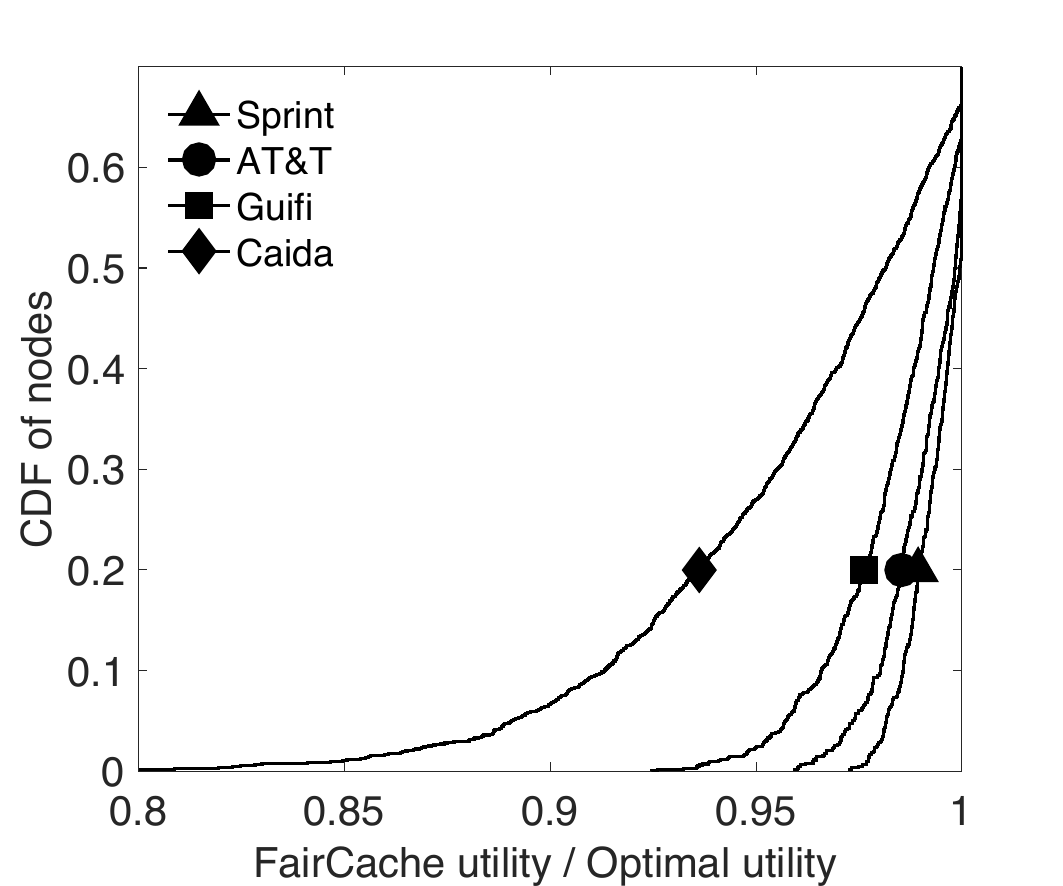}}
  \caption{Compared to the optimal algorithm, FairCache is more scalable on both real and synthetic networks. FairCache has a faster convergence rate and generates less traffic overhead than the optimal. Meanwhile FairCache achieves high accuracy.}
\label{fig:convergence}
\end{figure*}

We start by exploring scalability, measured by FairCache's convergence rate, \ie how many rounds of message exchange it takes the algorithm to bootstrap. This happens once at initiation: future dynamics are addressed using periodic low cost updates that are algorithmically trivial. Figure~\ref{fig:convergence:3} compares the convergence rates of the optimal (the upper figure) and FairCache (the lower figure) on the various topologies. The optimal needs significantly longer time to bootstrap than FairCache. Unsurprisingly, larger cache sizes also lead to a longer convergence time, as more state must be exchanged. To investigate how network size impacts the convergence rate, we use synthetic topologies with $4$~GB caches. The lines in Figure~\ref{fig:convergence:2} are clearly divided into two groups: the upper one is the optimal (with hollow markers) and the lower one is FairCache (with filled markers). The convergence rate degrades as the network size grows. Importantly, though, the increase in convergence time only grows sub-linearly, stabilising at networks of size 1k; we experimented with topologies of up to 9k nodes to find consistent results.



We also measure FairCache's scalability by its traffic overheads. Clearly, it is undesirable to generate large amounts of control messages to bootstrap. In this experiment, we measure the aggregated size of control messages for both the distributed optimal and FairCache as $C_O$ and $C_F$ respectively. We then calculate the traffic reduction as $\frac{C_O - C_F}{C_O}$. Figure~\ref{fig:cost:1} presents a box plot of the results for both $BA_1$ (upper boxes) and $ER_1$ (lower boxes) topologies. It shows that FairCache is able to achieve over 80\% traffic reductions, even on small networks of 100 nodes. As the network size increases, the benefit of using FairCache becomes more obvious. In a network of 900 nodes, FairCache attains 95\% reductions. This equates to significant traffic volumes; in one iteration, a 500 node network with $10^3$ objects can save 887 MB of control traffic via FairCache (leaving only 66.8 MB). With FairCache, on average, each cache only introduces 136~KB traffic overhead in an iteration.  We can therefore combine the above message overhead and convergence measurements to calculate the convergence time. If we configure the rate of control messages to 100~KB/s, FairCache takes 11 minutes to bootstrap. This is just 4.6\% of the time taken by the distributed optimal algorithm. Given a saturated 54~Mbit link, the FairCache control messages would therefore consume just 1.4\% of \highlight{bandwidth}. Importantly, this is only a bootstrap process; changes in request patterns are addressed with low cost updates within each node's neighbourhood.
\highlight{Even in highly dynamic situations where demands are volatile, we can let nodes exchange demand information in the background while running FairCache. Once the new demand matrix is constructed, FairCache simply re-calculates the new solution. Recall our setting in Section~\ref{sec:complexity} where each node needs to maintain $651$~KB demand information. To guarantee FairCache continuously runs, a total of $651 \times 10$~KB needs to be exchanged among $10$ nodes within $11$ minutes, these updates constitute under $10$~KB/s. These values can be configured to reflect the operating environment (we anticipate that in many cases operators would transfer state at much higher rates). Overall, we believe these overheads are more than acceptable for the overall performance gains (\S\ref{sec:exp:per})}

\subsection{Accuracy}
\label{sec:exp:vs}

FairCache significantly reduces the convergence time and messaging overhead of fairly allocating caching responsibilities. These improvements potentially come at the cost of accuracy (\ie lower utility than optimal solution). We next inspect the accuracy sacrifice required to obtain these improvements.

To measure the accuracy of FairCache, we compare it against the optimal algorithm using multiple topologies of different sizes. We first run the optimal algorithm and measure the utility $U_i$ for every node $i$. Similarly, we run FairCache and measure the utility $U_i'$. We then calculate the accuracy of FairCache as its ratio to the optimal for every node, \ie $\frac{U_i'}{U_i}$. Figure~\ref{fig:accuracy} plots the per node CDF of this ratio. We can see that FairCache achieves very high accuracy. For large networks like Guifi, all the nodes achieve an accuracy of over 92\%. For medium size networks like Sprint, all the nodes have at least 97\% accuracy and about 50\% of the nodes reach 100\% accuracy. \highlight{Besides Figure~\ref{fig:accuracy}, we also measure the \emph{aggregated} accuracy using $\frac{\sum_i U_i'}{\sum_i U_i}$. We find it is always above $95\%$ for medium-sized networks, whilst it decreases slightly for larger networks (\ie 3\% drop from Sprint to Guifi). Even for the very large CAIDA topology (over $25$k nodes), the average accuracy is still $\approx 95\%$ with fewer than $6\%$ nodes that have an accuracy drop of $10\%\sim20\%$.} These findings are consistent across the other topologies. To validate that these benefits continue to be enjoyed by large topologies, we also repeated the experiments presented in Figure~\ref{fig:convergence:2} on topologies ranging from 1k--9k nodes. Again, we find high levels of accuracy, stabilising at 95\%. 


The results confirm the rationale behind the heuristics used by FairCache. The approximation introduces almost negligible degradation in the accuracy. The main reason is that the highly skewed content popularity means that the bulk of caching decisions are limited to the most popular objects. This means that FairCache can attain high accuracy without requiring to share information about all objects (unlike the optimal). Further, by localising interactions to neighbouring nodes, FairCache can scale-up easily, without be overly affected by increasing network sizes.

\subsection{Price of Fairness}
\label{sec:exp:pof}

\begin{figure*}[!htp]
  \centering
  \subfloat[PoF on synthetic networks.]{\label{fig:perf:3}\includegraphics[width=4.5cm]{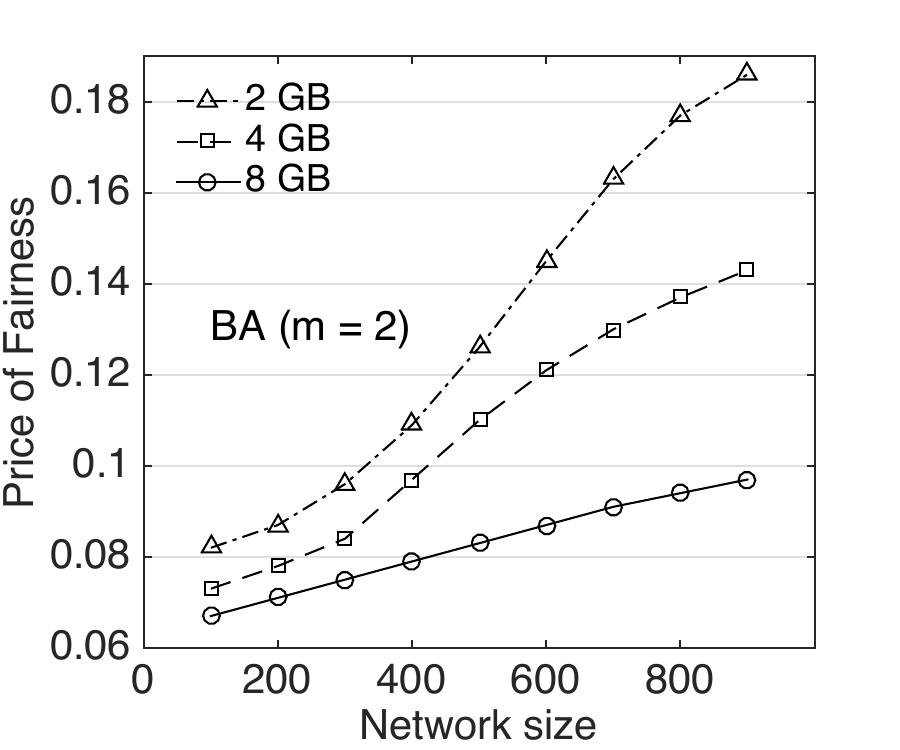}}
  \subfloat[PoF on ISP networks.]{\label{fig:perf:4}\includegraphics[width=4.5cm]{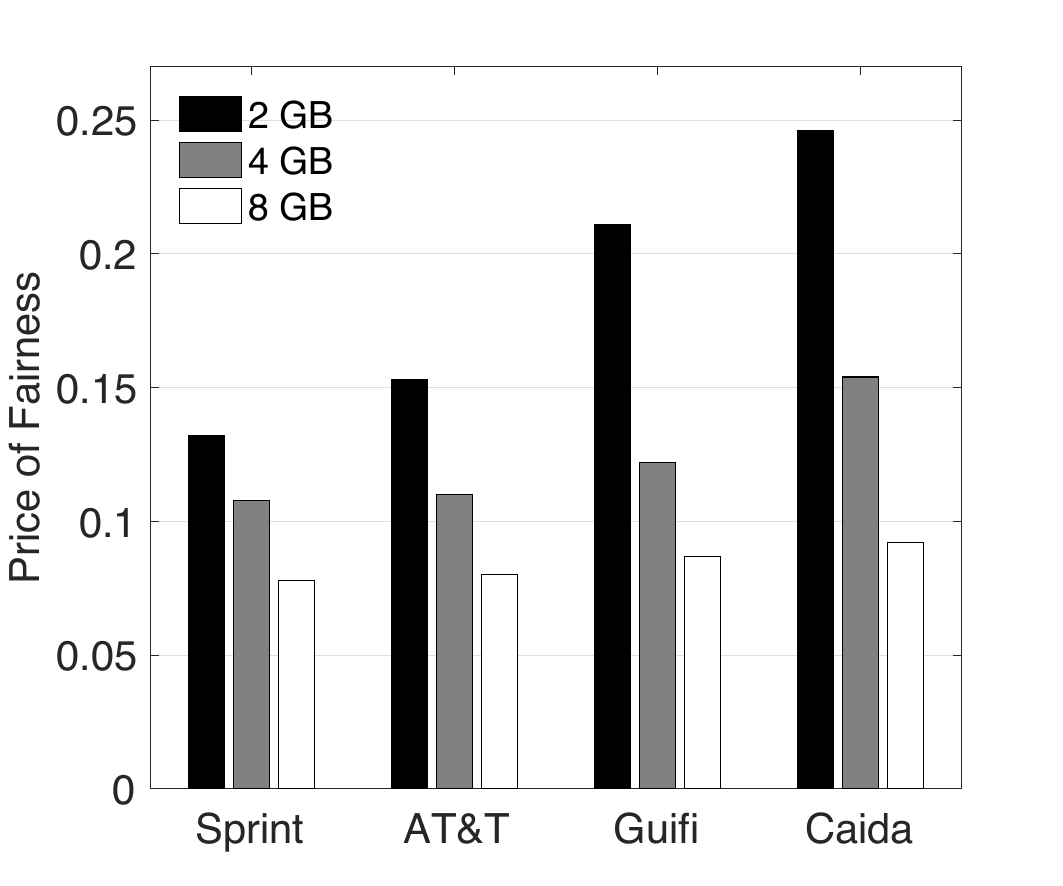}}
  \subfloat[Byte hit rate, 4GB cache.]{\label{fig:perf:1}\includegraphics[width=4.5cm]{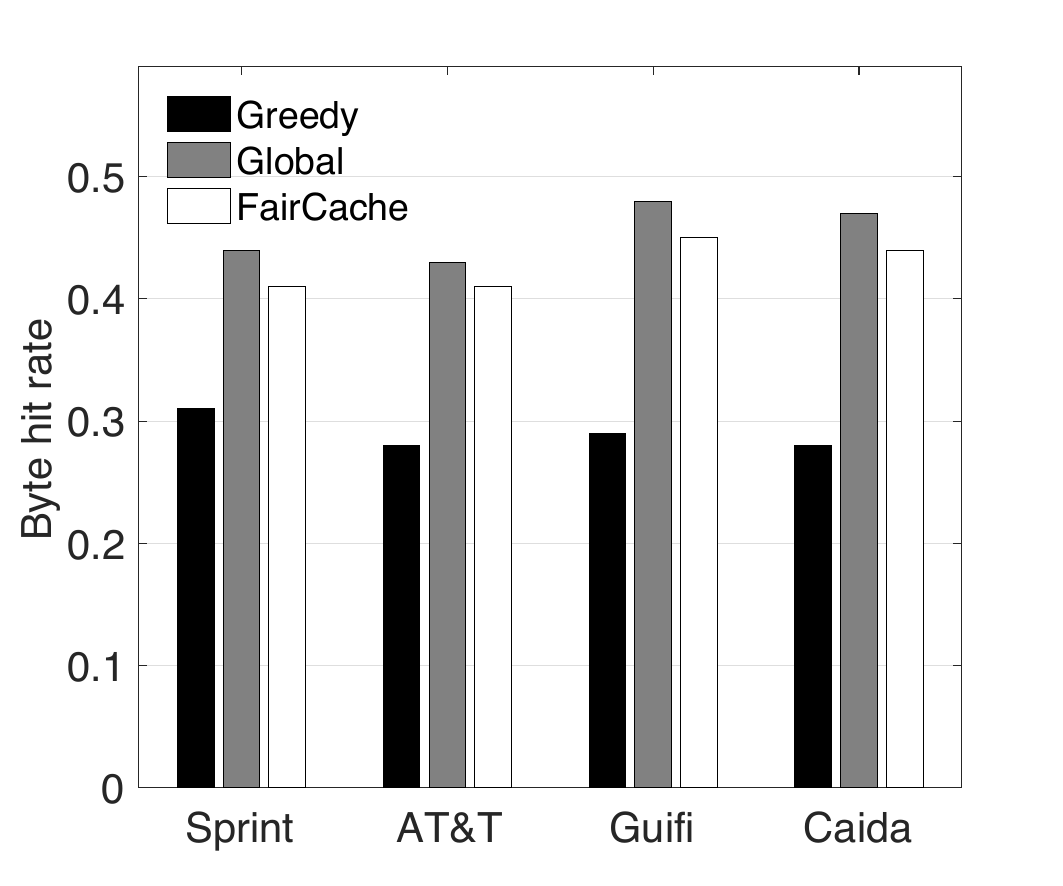}}
  \subfloat[Footprint reduction, 4GB cache.]{\label{fig:perf:2}\includegraphics[width=4.5cm]{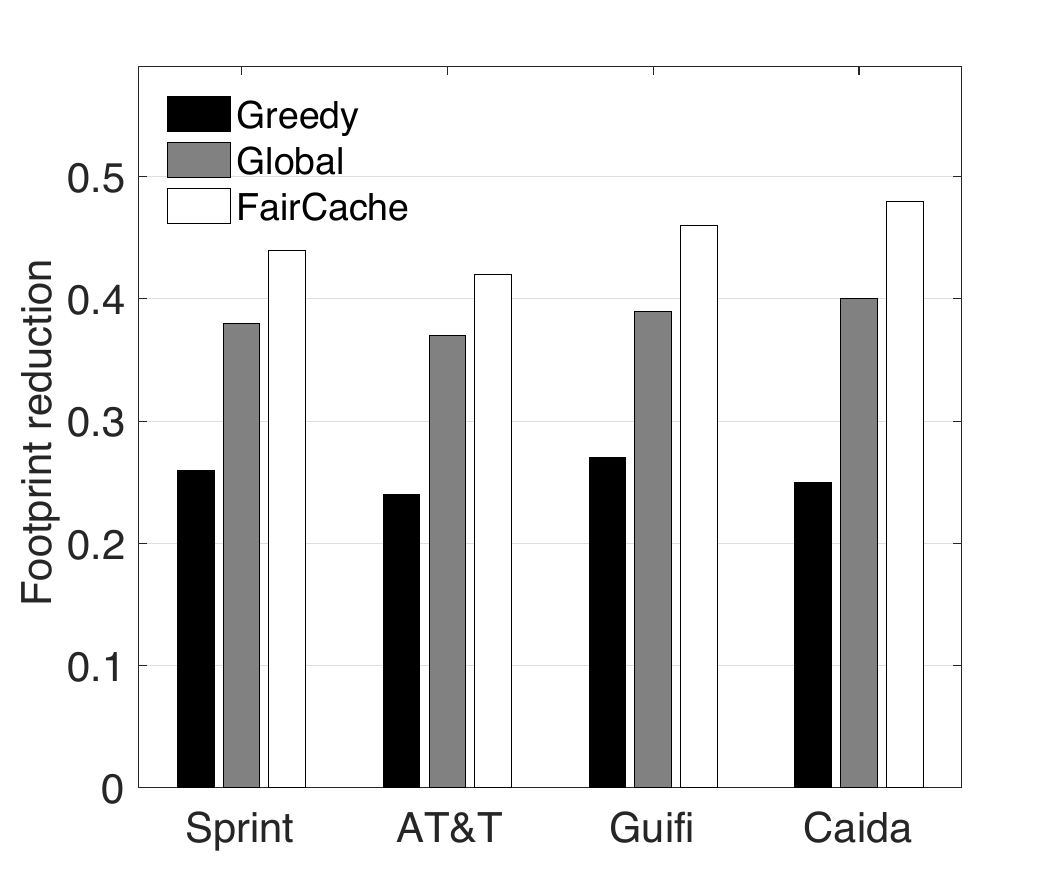}}
  \caption{FairCache achieves fairness by trading off some efficiency. However, a large cache size can effectively reduce PoF. In reality, FairCache is able to achieve very similar performance as Global, and is superior to Greedy in all cases.}
  \label{fig:convergence}
\end{figure*}


FairCache aims to realise fair collaboration amongst nodes, which could cause a degradation in aggregated global utility. We use the \textit{Price of Fairness} (PoF) to measure the loss in utility. \highlight{The PoF is calculated as the ratio between the aggregated utility loss of all nodes using FairCache and the global optimal that does not consider fairness~\cite{bertsimas2011price}}. A higher PoF value indicates a larger utility sacrifice.

Figure~\ref{fig:perf:3} and \ref{fig:perf:4} plot the PoF results of using both real and synthetic networks with three cache sizes. Both figures convey the same information, which is that the PoF increases as network size increases. 
We experiment with both realistic and synthetic networks of up to 9k nodes (Table \ref{tab:scalability} summarises the results), to find that the PoF stabilises after reaching a size of $\approx$3k nodes with a maximum PoF of 24\% (and a maximum of 20\% in the real topologies, \eg~Guifi). This is not negligible, but is likely not significant enough to dissuade caches that are interested in fairness from using FairCache.
Interestingly, our results also show that increasing the cache size is an effective way to ameliorate the loss in efficiency. In Figure~\ref{fig:perf:3}, a $4$~GB cache significantly improves the PoF. Using a $2$ GB cache, the PoF increases by 11\% when the network size increases from 100 to 900, whereas the PoF only increases by 3\% if a $4$ GB cache is used. Figure~\ref{fig:perf:4} shows similar properties with, for example, a 57\% improvement in PoF when increasing the cache size from $2$ GB cache to $4$ GB in Guifi. 

Overall, we believe that an average PoF of $<8\%$ is a cost worth paying for those concerned by a need for fairness. Moreover, FairCache exhibits good scalability regarding both accuracy and PoF, as the results in Table \ref{tab:scalability} show. The accuracy only slightly degrades (1.5\%) from 1000 nodes to 9000 nodes, and is always above 94.7\%. The accuracy stabilises after 2000 nodes. Similarly, PoF stabilises after reaching around 3000 nodes. For the 2GB cache configuration, PoF is capped by 24.3\%; 4GB by 15.6\%; 8GB by 11.4\% (not included in Table \ref{tab:scalability} due to space limit). 
The scalability of FairCache can be explained as follows: only small neighbourhoods play an important role in deciding a node's overall performance. As the size of the neighbourhood is relatively unaffected by the network size, this property always exists.




\begin{table}[!tb]
  \centering
  \begin{tabular}{  l  l  l  l  l  l  }
  \multicolumn{6}{c}{Scalability test on PoF and accuracy using real ISP topologies.} \\
    \hline
    ASN & \#nodes & \#edges & Accuracy & PoF (2GB) & PoF (4GB)  \\
    \hline
    \#1755 & 295   & 544     & 97.3\%   & 16.5\%  & 11.3\% \\
    \#3356 & 1620 & 6743   & 96.2\%   & 20.8\%  & 13.7\% \\
    \#1221 & 2669 & 3181   & 95.4\%   & 22.6\%  & 14.9\% \\
    \#2914 & 4670 & 7618   & 94.8\%   & 23.9\%  & 15.4\% \\
    \#1239 & 7337 & 9924   & 94.7\%   & 24.3\%  & 15.6\% \\
    \#7018 & 9430 & 11682 & 94.8\%   & 24.1\%  & 15.2\% \\
    Caida & 25107 & 49458 & 94.6\%   & 24.4\%  & 15.4\% \\
    \hline
  \end{tabular}
  \caption{\highlight{Key metrics for increasing network sizes (295 - 25k nodes). Networks are real ISP topologies taken from Rocketfuel and CAIDA AS-level trace.}}
  \label{tab:scalability}
\end{table}

\subsection{Caching Performance}
\label{sec:exp:per}

The previous section has shown that utility is reduced by considering fairness. Next, we explore performance from the perspective of traditional metrics: byte hit rate and footprint reduction. Hit rate is a conventional metric to measure saving on inter-domain traffic, whilst footprint reduction is the reduction on the product of traffic volume and distance.

We compare FairCache against two other strategies: \one~\emph{Greedy}, which computes the local optimal for each cache without collaboration; and \two~\emph{Global}, which maximises the aggregated utility. Figure~\ref{fig:perf:1} and \ref{fig:perf:2} plot the results on the real networks using $4$ GB caches. Naturally, Figure~\ref{fig:perf:1} shows that Global achieves the best hit rates due the fact that it optimises the overall network. That said, FairCache only performs slightly worse, with a 5\%--10\% performance degradation. Compared to Greedy, FairCache is consistently superior with at least a 28\% improvement. This shows that, regardless of fairness, FairCache can offer significant performance improvements over local algorithms (note that Greedy is the theoretical upper bound of algorithms such as Least Recently Used).
When inspecting the traffic footprint reduction, performance is even higher. FairCache is superior in all networks. Although the reasons are intuitive for Greedy, which sees nodes locally optimising, it is more surprising in Global. The reason is that FairCache only requests from nearby caches (limited by $r$). In contrast, Global uses \emph{any} node in the network. This increases hit rates, but results in more traffic.




To have a closer look how utility is spread across caches, we select the AT\&T network and study the utility distribution in the network (\ie how are traffic savings distributed across caches). Figure~\ref{fig:hit:1:1} plots the CDF of normalised utility values across each node (normalised by the top value per simulation). By comparing Greedy and FairCache, we see that \emph{every} node is better off through collaboration using FairCache (note this is also the case across all other topologies and cache sizes). On the other hand, the Global strategy intersects with both Greedy and FairCache, \ie some caches in Global get lower utility than Greedy. The area between the lines indicates the percentage of caches that are worse off due to global optimisation. The Global strategy leads to 13\% of nodes getting worse off compared to Greedy, and 20\% compared to FairCache. With Global, these nodes should rationally cease to cooperate. Again, regarding the aggregated utility, Global is only about 5\% better than FairCache. Moreover, the CDF curve of Global is more stretched than that of FairCache, which indicates there are much larger variations in nodes' utilities when using the Global strategy, \ie benefits are not evenly distributed.

\begin{figure}[!htp]
  \centering 
  \subfloat[Cumulative distribution of utilities.]{\label{fig:hit:1:1}\includegraphics[width=4.5cm]{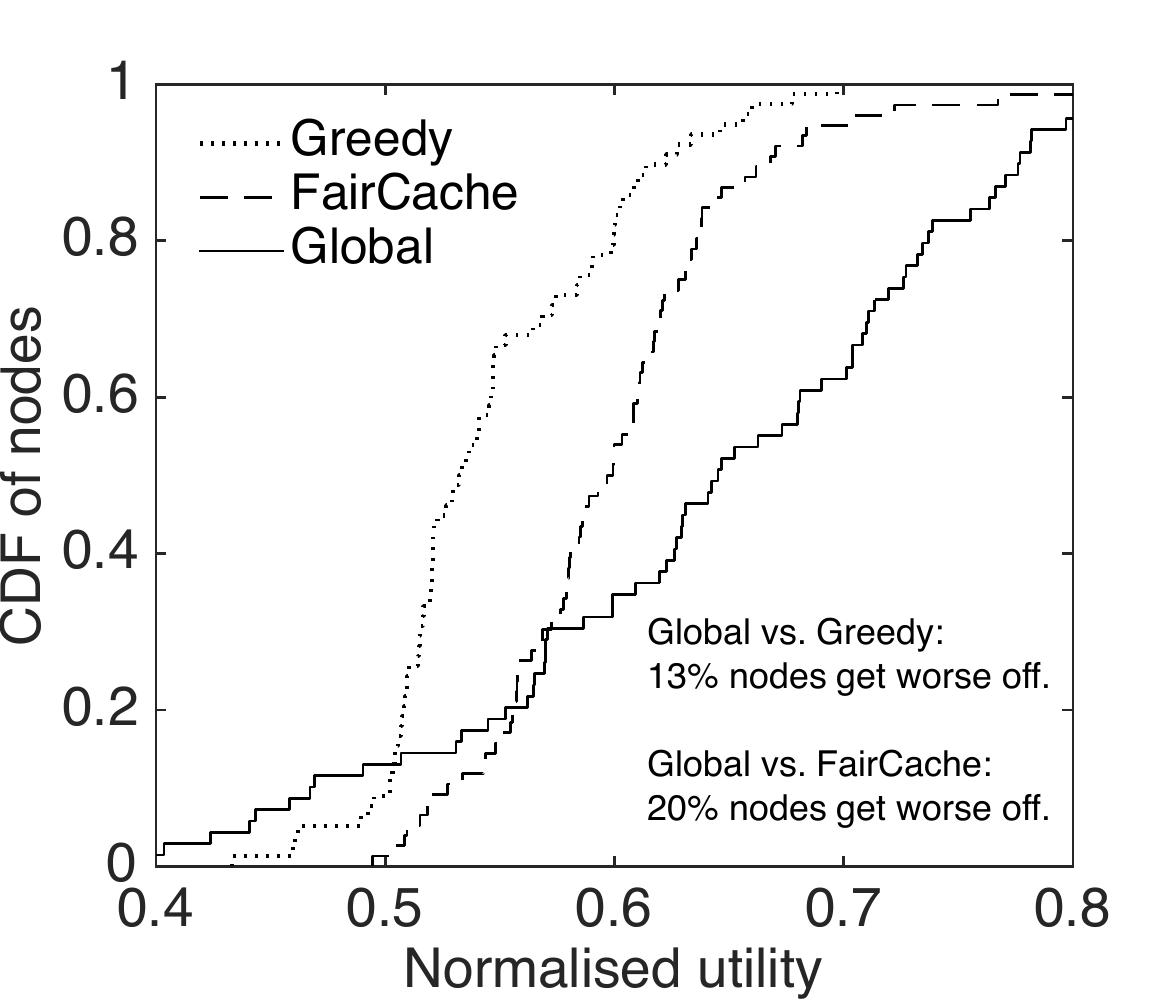}}
  \subfloat[Betw. centrality vs. utility.]{\label{fig:hit:1:2}\includegraphics[width=4.5cm]{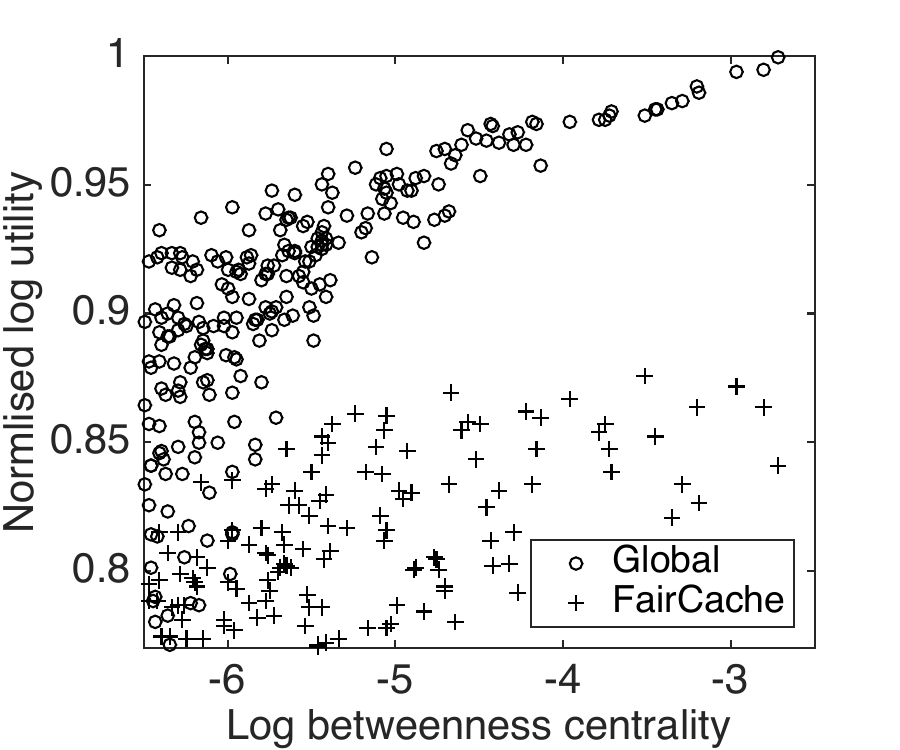}}
  \caption{Comparison of strategies on AT\&T, 4~GB cache size.}
  \label{fig:hit:1}
\end{figure}

Figure~\ref{fig:hit:1:2} shows the log-log plot of nodes' normalised utility as a function of betweenness centrality~\cite{6739053, Chai:2012:CLM}. Nodes with a high betweenness are core routers, whilst those with low betweenness are usually found at the edge. Interestingly, when nodes use the Global strategy, a node's utility strongly correlates with its position in the network: core nodes gain the highest utility. This is because the Global optimal tends to place all the popular (\ie high value) content at the core to reduce duplicates --- a theoretically attractive, but practically infeasible approach. In contrast, FairCache significantly weakens this correlation. This is beneficial as it means that utility is also increased at the edge caches. As well as improving fairness, it also reduces load in the backbone and provides consumers with lower delay access to object. This also contributes to FairCache's high traffic reductions, as hits are pushed closer to clients.

\subsection{Sensitivity Analysis of Spatial/Content Locality}
\label{sec:sens}

\begin{figure*}[!htp]
  \centering 
  \subfloat[Byte hit rate vs. spatial locality]{\label{fig:guifi:1}\includegraphics[width=4.5cm]{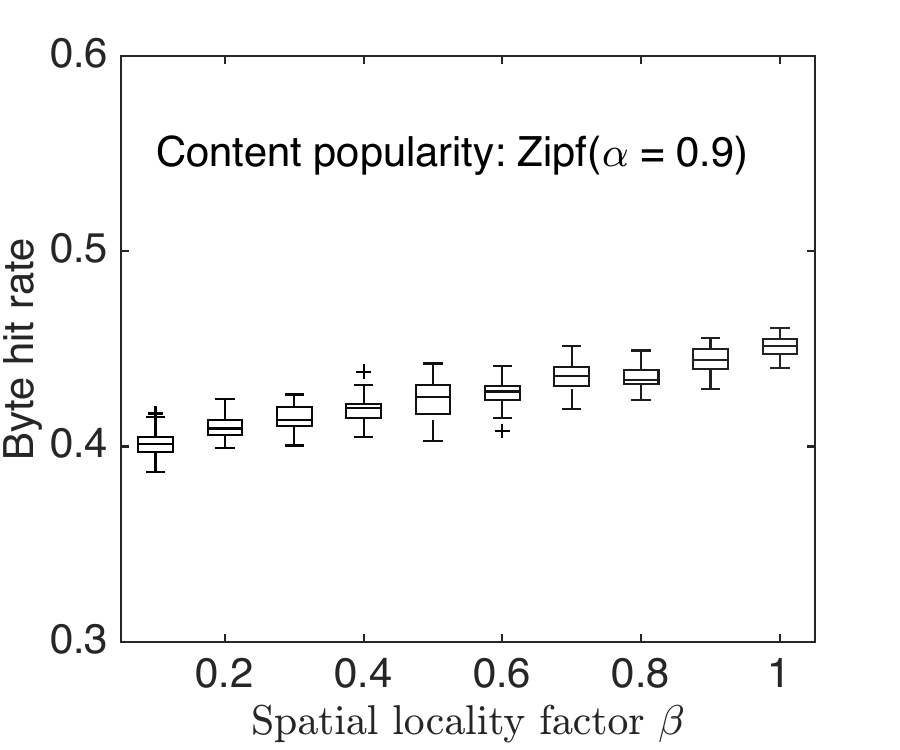}}
  \subfloat[FP reduction vs. spatial locality]{\label{fig:guifi:2}\includegraphics[width=4.5cm]{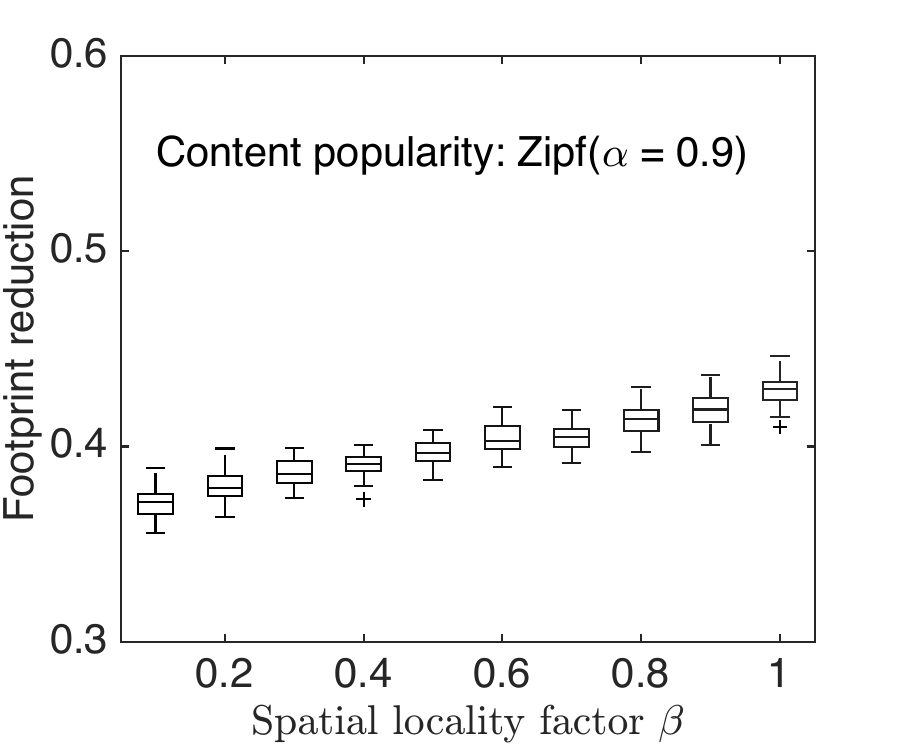}}
  \subfloat[Accuracy vs. popularity skew]{\label{fig:guifi:3}\includegraphics[width=4.5cm]{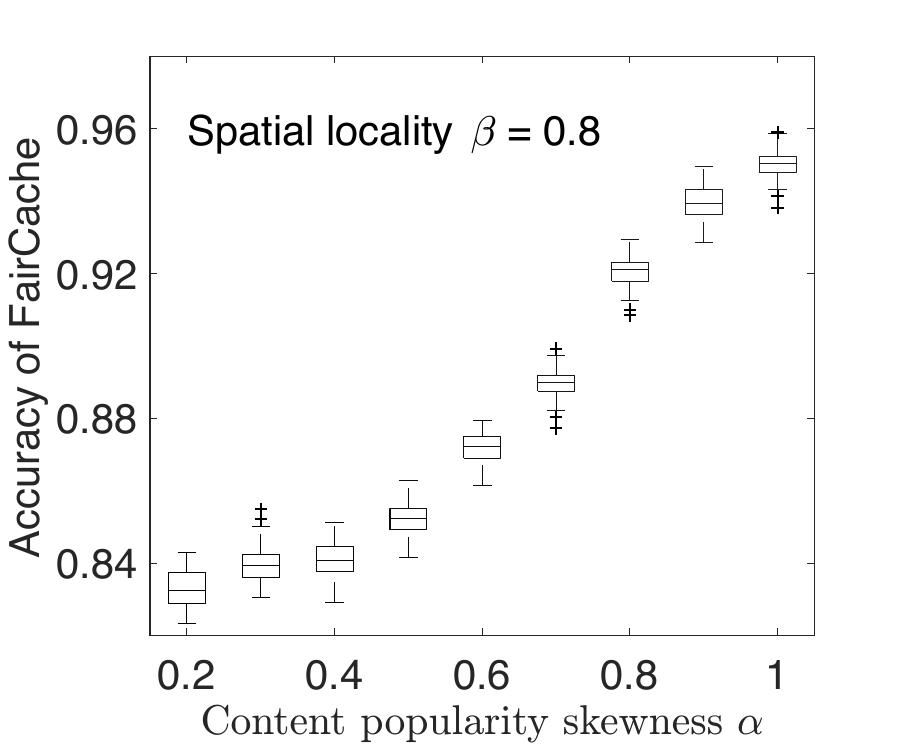}}
  \subfloat[Msg reduction vs. popularity skew]{\label{fig:guifi:4}\includegraphics[width=4.5cm]{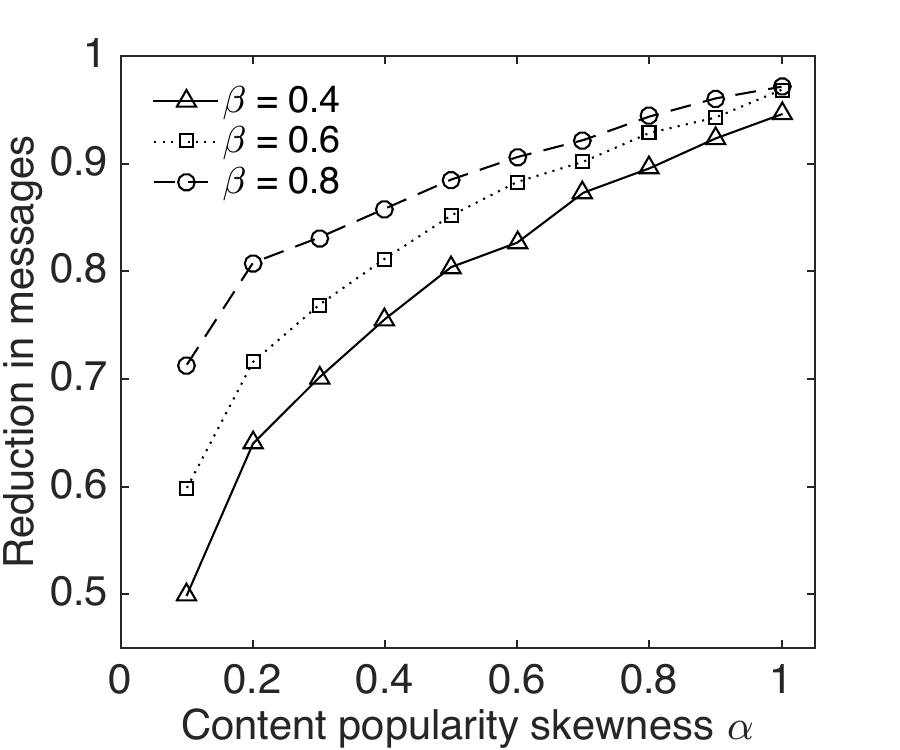}}
  \caption{Experiments on the Guifi network, 4~GB cache size. We vary both content popularity skewness $\alpha$ and spatial locality factor $\beta$ from $0.1$ to $1$. We observe a gradual and slow improvement in caching performance as spatial locality factor increases. Both spatial locality and content popularity skewness have significant impacts on the accuracy and the traffic reduction.}
  \label{fig:guifi}
\end{figure*}

FairCache's heuristics take advantage of highly skewed spatial and content popularity localities.
A natural question is how these localities impact the algorithm. To explore this, we perform sensitivity analysis across these two parameters to measure the robustness of our heuristics. Here, we solely present the Guifi topology due to space constraints. The reason we select Guifi is that the dataset contains geographic coordinates of each node, allowing much more fine grained analysis of spatial locality. We have confirmed that the results are representative of the other topologies.

We use a Hawkes process-based algorithm~\cite{Dabirmoghaddam:2014:UOC:2660129.2660143, Wang:2015:PUS:2810156.2810162} to generate a user request trace. The algorithm is controlled by two parameters: a \textit{content popularity skewness} $\alpha$ and a \textit{spatial locality factor} $\beta$. $\alpha$ controls the overall content popularity which follows \textit{Zipf}$(\alpha)$. The spatial locality factor, $\beta=0$, means  the request pattern reduces to an \textit{Independent Reference Model}; whilst $\beta=1 $ indicates very high spatial localisation (\ie requests for an object often occur in the same locale). 


First, we inspect their impact on the caching performance metrics. Figure~\ref{fig:guifi} presents the results by varying both $\alpha$ and $\beta$ in $(0,1]$. From Figure~\ref{fig:guifi:1} and \ref{fig:guifi:2}, we observe a shallow improvement on byte hit rate and footprint reduction as $\beta$ increases. Specifically, they increase by only $6\%$ and $8\%$ respectively when increasing $\beta$ from $0.1$ to $1$. This suggests that spatial locality is not a critical requirement for FairCache. 


On the other hand, the popularity skew, $\alpha$, has a more significant impact on the accuracy and message reduction of FairCache. Figure~\ref{fig:guifi:3} shows that the average accuracy of FairCache improves from $85\%$ to $97\%$ by increasing $\alpha$ from $0.2$ to $1$. The speed of degradation of accuracy by decreasing $\alpha$ also slows down at certain point ($\alpha=0.4$). The reason is because the general popularity distribution gets closer to a uniform distribution (due to a small $\alpha$). Thus, items are randomly requested, which means that each object has a similar utility when being cached. Interestingly, this means the overall utility of a cache will not vary much, though the solution can be quite different from the optimal one.


Last, we inspect the messaging overhead of running FairCache, presented as the reduction in comparison to the distributed optimal solution again. In Figure~\ref{fig:guifi:4}, we see that both $\alpha$ and $\beta$ have a notable impact. Higher $\alpha$ and $\beta$ both result in lower overheads (\ie higher reductions). The reason is that a smaller $\alpha$ value leads to a more uniform popularity distribution, which makes the demand matrices deviate more from each other, which further leads to larger exchanged messages for $\boldsymbol{\lambda}$ values. The smaller $\beta$ values have almost the same effect on demand matrices as that of $\alpha$. However, we also notice that $\beta$ has more significant impacts when $\alpha$ is small.


\subsection{Neighbourhood Size Distribution}
\label{sec:nbhd}

\begin{figure}[!htp]
  \centering 
  \subfloat[Avg. radius vs. avg. degree]{\label{fig:radius:1:1}\includegraphics[width=4.5cm]{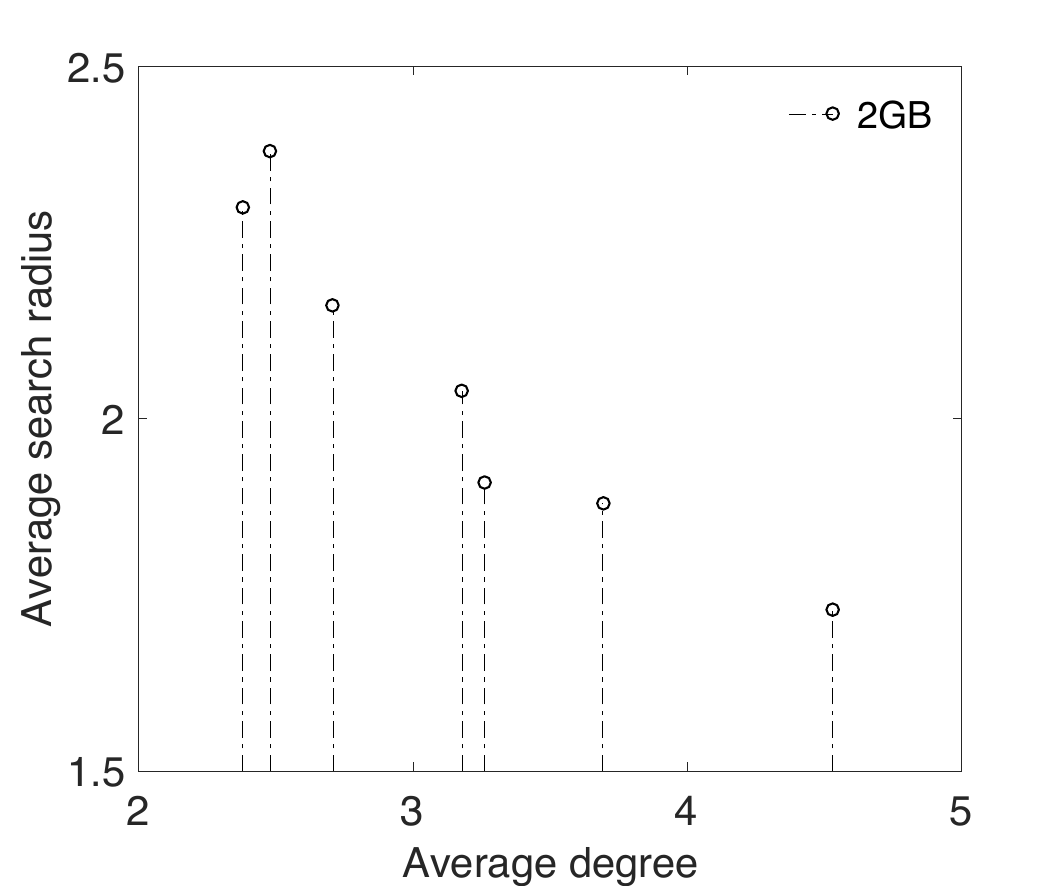}}
  \subfloat[Radius vs. cache size heatmap]{\label{fig:radius:1:2}\includegraphics[width=4.5cm]{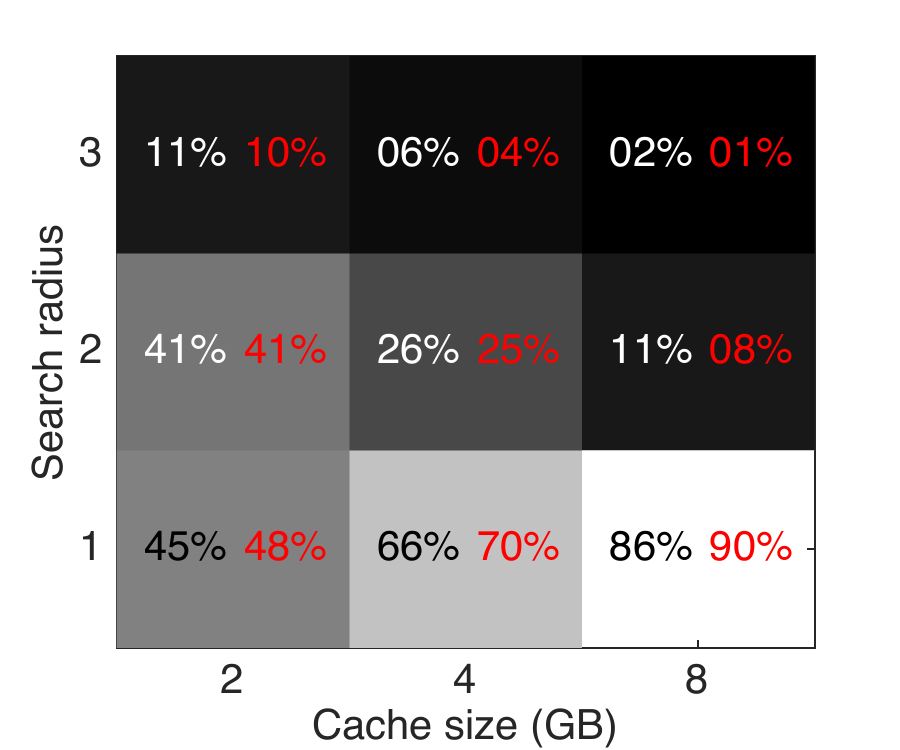}}
  \caption{The distribution of neighbourhood size (in terms of $r$) after FairCache converges. $(a)$ uses seven ISP topologies; there is a relatively strong negative correlation between nodes' average degree and their average neighbourhood size. $(b)$ uses Guifi topology with three cache configurations; the numbers in the grid show the percentage of the nodes and percentage of the control messages (in red).}
  \label{fig:radius:1}
\end{figure}

As previously stated in \S~\ref{sec:complexity}, maintaining some neighbourhoods is critical for ensuring scalability. This is because \one~it dictates the communication overhead; and \two~it justifies the effectiveness of the heuristic~\three. In this section, we present our empirical study on the neighbourhood size. We launch a number of experiments using the default setup, whilst varying the topology (as this is what dictates the neighbourhood size). We measure the neighbourhood sizes after FairCache converges (represented by the search radius $r$).

Figure \ref{fig:radius:1:1} plots the average search radius, $r$, as a function of the average degree of the topology. We use seven $r1$-level ISP topologies \highlight{(\ie ISP router-level topology with one-hop clients included)}, with a 2GB cache size. As the average degree increases, $r$ decreases, thereby reducing network overhead. This is because network density increases, it becomes less necessary to create multi-hop neighbourhoods. Most important is the fact that even with a low degree, the search radius is very small on all topologies (around 2 hops).

Figure \ref{fig:radius:1:2} presents results for experiments performed on the Guifi network. It shows a heatmap, which reveals the percentage of nodes that have a certain radius (across three cache sizes). Each point in the grid shows the percentage of nodes in a simulation that have a certain search radius (lighter colour means more), as well as the percentage of control messages generated within each neighbourhood (shown in red). Figure \ref{fig:radius:1:2} conveys two pieces of important information. First, most of the nodes end up with a neighbourhood of fewer than 3 hops. Even for the small 2GB cache sizes, $87\%$ of the nodes have no more than 2-hop search radius. Second, using larger cache size makes the distribution even more skewed, which further leads to even smaller search radius. When we increase the cache size from 2GB to 8GB, the percentage of nodes with a 1-hop radius increases from $45\%$ to $86\%$. With 8~GB cache, $97\%$ of nodes' final neighbourhoods are no more than 2 hops. In addition, we also provide the distribution of control messages in each $r$-hop neighbourhood, plotted in red colour in the same grid. As we can see, the traffic distribution is even more condensed within a very small neighbourhood, which indicates most interactions are between directly connected neighbours. We have confirmed that these results are mirrored across all other topologies.


\section{Discussion and Limitations}
\label{sec:discussion}

Deployment of FairCache raises a number of interesting questions. A key practical concern is the potential for parties to participate in FairCache in a malicious or non-cooperative manner. This is a possibility as nodes are expected to report shadow prices, which they could be manipulated. FairCache is not intended to force stakeholders to collaborate, or to protect against cheating; hence, we have assumed that all nodes adhere to the FairCache algorithm. However, if deployed, such complexities would need to be addressed. In current inter-domain network protocols this problem is handled using out-of-band trust establishment, alongside signing authorities to bind autonomous systems to trusted identities (\eg~RPKI\cite{Wahlisch:2015:RTS:2834050.2834102}). Equally, we envisage FairCache could rely on similar principles, in which legally formed (potentially transitive) collaboration agreements are underpinned by public key cryptography. 

There are also a number of alternative practical concerns. 
\highlight{We implicitly assumed that the demand matrix is stationary, whereas in reality the demand will change over time. Recall that the space complexity (in Section \ref{sec:complexity}) of storing an individual demand matrix in FairCache is $\Theta(|O'| |\overline{N}|)$, hence the complexity of updating the whole matrix is bounded by $\Theta(|O'| |\overline{N}|)$. Two other facts further help to reduce this overhead: \one~the spatial locality we have mentioned; and \two~the various research \cite{Cha:2007:ITY:1298306.1298309} that has reported that such changes are gradual. Hence we can incrementally update the demand matrix to avoid unnecessary traffic. Furthermore, such occasional and incremental updates do not necessarily need to be synchronised especially when FairCache is deployed as an ever-running background process. In theory, the stale information may slow down the convergence in an iterative optimisation process. In practice, less-frequent and incremental updates can ameliorate such impacts\cite{boyd2004convex}. Moreover, due to the nature of FairCache algorithm, the updates mostly affect the local neighbourhood and their cascading effects will drastically decrease out of the neighbourhood.}

\gchange{Although we have focussed on building an efficient algorithm, we also acknowledge that FairCache treats storage as the key bottleneck. There are also a number of other constraints that could be included~\cite{perino2011reality}. For instance, hardware bottlenecks can render servers useless even whilst in possession of content (\eg CPU, I/O bus, congestion collapse). Thus, deployment would probably involve the introduction of such considerations into our model of fairness and utility. This could, for example, result in caches actively storing the \emph{same} object in an attempt to share heavy load.} Another simplifying assumption is the modelling of delay using hop count (like BGP); whereas this is a useful abstraction, it does not consider the variability introduced by realtime congestion. We consider this an acceptable sacrifice, as introducing such realtime considerations would introduce burdensome overheads. Another point worth highlighting is that we base caching decisions on per-object popularity counts, therefore introducing greater memory overheads than algorithms like Least Recently Used. We emphasise, however, that our heuristic removes all unpopular content, making such counts highly feasible.

\section{Related Work}
\label{sec:related}

There are three key related areas of work: collaborative caching, content delivery networks (CDNs) and game theoretical studies of caching. Collaborative in-network caching has been proven as an effective methodology to improve system performance in various contexts~\cite{Dahlin:1994:CCU:1267638.1267657, Fan:2000:SCS:343571.343572, 5062201, 6566743, Chai:2012:CLM, Psaras:2012:PIC:2342488.2342501, 6195634, borst:DistributedCaching_INFOCOM2010}, even though edge caching has also been shown to be effective~\cite{Fayazbakhsh:2013:LPM}.
Previously proposed solutions are either limited by a centralised solver \cite{Dahlin:1994:CCU:1267638.1267657, Chun:2004:SCD} which makes scalability difficult, or limited by distributed heuristics \cite{5062201, Chai:2012:CLM, Psaras:2012:PIC:2342488.2342501, 6566743, 6195634, borst:DistributedCaching_INFOCOM2010, Fan:2000:SCS:343571.343572}, which neither guarantees a global optimum nor Pareto efficiency. FairCache is most related to the latter in that we do not guarantee a global optimum; however, we build on their contributions by introducing the concept of fairness and ensuring Pareto efficiency. Importantly, we also reveal the need for fairness to encourage engagement by cache operators.
\highlight{Recently, fair cache sharing also attracts enough attention in cloud computing and system research, \eg \cite{Pu:2016:FNF:2930611.2930637} proposes FairRide using blocking to achieve isolation-guarantee and strategy-proofness properties.}

The current solution used for Internet-scale content delivery are CDNs. They hold many similarities to ICNs~\cite{Fayazbakhsh:2013:LPM}, however, unlike our proposal, they are not collaborative entities. Typically, they are operated by distinct companies that deploy independent infrastructures. Some, like Akamai, sell their capacity to third party content providers (arguably a form of collaboration), whilst others build dedicated infrastructures for their own content (\eg Google, Facebook, Netflix). Recent work within the IETF has endeavoured to support inter-CDN cache sharing~\cite{rfc7336}, however, this only provides protocol support, rather than algorithms to decide when, where and how caches should be shared. Hence, our work is orthogonal, and could be applied to CDNs \highlight{as a recent proposal in \cite{7218624}}.

Game theory is an effective tool to analyse the effects of individual behaviours in a complex system;
\eg prior work \cite{Key:2011:PSM:1866739.1866762} analyses the fairness achieved in bandwidth allocation by coordinated and uncoordinated rate control over multiple links in peer-to-peer networks.
Recent work~\cite{Chun:2004:SCD, Pacifici:2011:scr, Pacifici:2011:cntp, 1717403, 5686876, 7524445} applies game theory to study in-network caching. In these papers, the caching problem is modelled as non-cooperative, pure strategic games and the equilibrium is analysed. Unlike us, these work take a system-level utilitarian approach that aims to achieve a global optimum. In contrast, we focus on attaining fairness amongst nodes. More related to us is \cite{Pacifici:2011:scr, 1717403, Chun:2004:SCD}, which look at how selfishness drives nodes to act. These studies show how selfishness impact the equilibrium and efficiency in cache systems (measured by the \textit{Price of Anarchy}). They also show that the global optimum is seldom achieved due to lack of coordination and nodes' inherent selfishness. Again, fairness is overlooked though; we introduce this as an integral requirement of cooperation. 

To the best of our knowledge, no prior work in ICN has tried to solve the collaborative caching as a bargaining game and has devised a low-complexity heuristic to embrace both efficiency and fairness.

\section{Conclusion}
\label{sec:conclusion}



To date, studies of collaborative ICN caching have focussed on traditional metrics such as hit rate, assuming that nodes are happy to contribute to achieving a global optimum. In this paper, we have argued that practical situations are unlikely to adhere to this model. Instead, caches operated by separate stakeholders will expect a reasonable level of \emph{fairness}, where they are not penalised for cooperating with others. We began by delineating an optimal solution, which ensures no node attains lower utility by collaborating. To address its high complexity, we have also proposed a heuristic algorithm, FairCache, which we have shown achieves high performance at a fraction of the cost. Unlike past work, FairCache offers Pareto efficiency and proportional fairness, ensuring that \emph{all} nodes are incentivised to collaborate.
As well as helping to promote cooperation, our results show that proportional fairness plays a key role in balancing network traffic too. It helps maintain more hits at the edge, rather than globally optimal solutions that centralise hits in the core. We are not prescriptive in how FairCache is deployed and have ensured that it can be used either globally or amongst a subset of collaborating nodes. Hence, our key take-home message is that future collaborative caching designs should cease to assume purely altruistic cooperation and, instead, be explicitly built around the concept of fairness.




There are several lines of potential future work. First, we plan to build a wireline protocol to implement FairCache's design. The most prominent challenge in this regard is implementing a FairCache protocol that is robust against cheating nodes. This is a fascinating area of future work; currently, we assume trusted certified parties that do not lie, however, we wish to expand this to cover more dynamic arrangements in which trust can be formed on-the-fly. Note that this would not necessarily involve significant changes to FairCache --- simply augmentary functions, \eg key exchange. \gchange{Clearly, this should be underpinned by a hardware implementation for exploring practical feasibility at line rates.}
There are various other real-world concerns that could also be integrated into FairCache too. For instance, dynamics regarding link availability, congestion and request patterns could be explored. This should extend to integrating new constraints (\eg bandwidth, power, CPU), \gchange{as well as alternate forms of fairness (\eg bandwidth fairness, user-centric fairness).} Lastly, we wish to expand FairCache to consider situations in which caches have external influences (\eg business arrangements) that modify their behaviours. As of yet, little work has considered exogenous incentives that drive caching collaboration. We therefore see this as a fruitful line of study.



\appendix
\label{apx:proof}
 
 Let $T: \mathbb{R}^n \rightarrow \mathbb{R}^n$ be the sorting operator used in \cite{bertsimas2011price}. More precisely, that is, we say $\mathbf{U}$ is lexicographically larger than or equal to $\mathbf{U}'$ if we can write $\mathbf{U} \succeq_\mathrm{lex} \mathbf{U}' \iff U_i \geq U'_i, \forall i \in [0, n)$ after applying $T$ to both $\mathbf{U}$ and $\mathbf{U}'$. $\succ_\mathrm{lex}$ can be similarly defined. If we let $\mathbf{U}^*$ be the corresponding utilities achieved by using $(\mathbf{x}^*,\mathbf{y}^*)$ which fulfils a certain well-defined fairness. Let $U_i^{w}$ denote the minimum utility of $v_i$ of a given solution. Technically, we have the following definitions of the two fairness metrics besides the previously defined Proportional Fairness (PF).
 
 \begin{mydef}
\label{def:effair}
Egalitarian Fairness (EF): $(\mathbf{x}^*,\mathbf{y}^*)$ is EF iff $\forall (\mathbf{x},\mathbf{y}) \neq (\mathbf{x}^*,\mathbf{y}^*) \Rightarrow$ the following two conditions cannot be both true at the same time. (1) $\exists i,$ s.t. $U_i \geq U^*_i$; (2) $\forall i,j$ s.t. $U_i -U_i^{w} = U_j - U_j^{w}$.
\end{mydef}

\begin{mydef}
\label{def:mffair}
Max-Min Fairness (MF): $(\mathbf{x}^*,\mathbf{y}^*)$ is MF iff  $\forall (\mathbf{x},\mathbf{y}) \neq (\mathbf{x}^*,\mathbf{y}^*) \Rightarrow T(\mathbf{U}^*) \succeq_\mathrm{lex} T(\mathbf{U})$.
\end{mydef}
 
 These definitions are followed by the theorems below, both of which have rather straightforward proofs as we will show in the following.
  
\begin{mythm}
\label{thm:4}
In a fair collaborative game $(\Omega,u^0)$, the optimal caching strategy $(\mathbf{x}^*,\mathbf{y}^*)$ achieves PF.
\end{mythm}

\begin{mythm}
\label{thm:5}
In a fair collaborative game $(\Omega,u^0)$ with optimal strategy $(\mathbf{x}^*,\mathbf{y}^*)$, EF is sufficient for MF, \ie EF $\Rightarrow$ MF.
\end{mythm}

In our caching games, \textit{PF} is naturally guaranteed by NBS as Theorem~\ref{thm:4} states. It is also intuitively easy to understand that if a Pareto efficient solution exists and achieves \textit{EF}, it also achieves \textit{MF} at the same time, as Theorem~\ref{thm:5} states.

 \subsection{Proof of Theorem \ref{thm:4}.}
 
 \begin{proof}
   Because $(\mathbf{x}^*,\mathbf{y}^*)$ is the optimal caching solution, namely $(\mathbf{x}^*,\mathbf{y}^*) = {\arg \max}_{\mathbf{x},\mathbf{y}} \sum_{v_i \in V} \ln (U_i - u^0_i)$. Let $f(\mathbf{U}) = \sum_{v_i \in V} \ln (U_i - u^0_i)$. For $f(\mathbf{U})$ to reach its maximum, the necessary and sufficient first order condition is $\nabla f^* = 0$. $\forall (\mathbf{x},\mathbf{y}) \neq (\mathbf{x}^*,\mathbf{y}^*) \Rightarrow \exists \boldsymbol{\lambda} \succ 0$ such that $\lambda_i^{-1}=U_i-u_i^0>0$. Then $\forall v_i \in V$ we have 
 \begin{align*}
   \nabla f^* - \boldsymbol{\lambda} \prec 0 & \Longrightarrow \frac{\partial f^*}{\partial U^*_i} - \lambda_i < 0 \\
   & \Longrightarrow \frac{1}{U^*_i - u^0_i} - \lambda_i < 0 \\
   & \Longrightarrow \frac{\lambda_i^{-1}}{U^*_i - u^0_i} - \frac{U^*_i - u^0_i}{U^*_i - u^0_i} < 0
 \end{align*}
 Sum over all the $v_i \in V$, we have
 \begin{align*}
   & \sum_{v_i \in V} \frac{(U_i - u^0_i) - (U^*_i - u^0_i)}{U^*_i - u^0_i} < 0 \Longrightarrow \sum_{v_i \in V} \frac{U_i - U^*_i}{U^*_i - u^0_i} < 0  
 \end{align*}
 By definition \ref{def:propfair}, strategy $(\mathbf{x}^*,\mathbf{y}^*)$ is proportionally fair.
 \end{proof}

 \subsection{Proof of Theorem \ref{thm:5}.}
 
 \begin{proof}
 Instead of directly proving that "\textit{EF} is sufficient for \textit{MF} in a fair collaborative game $(\Omega,u^0)$", we first prove that "\textit{EF} is sufficient for \textit{KS} fairness". Recall that \textit{KS} optimises the players of the worst utility, hence \textit{KS} fairness refers to the strategy that maximises the minimum utility in a game.

We prove the theorem by contradiction. Let's assume solution $(\mathbf{x}^*,\mathbf{y}^*)$ is egalitarian fair, but not \textit{KS} fair. $\mathbf{U}^*$ is the corresponding utility value.
 
 Let's further assume another solution $(\mathbf{x}',\mathbf{y}') \neq (\mathbf{x}^*,\mathbf{y}^*)$ which achieves \textit{KS} fairness, and $\mathbf{U}'$ is its utility value.
 In a fair collaborative game, based on the nature of Nash bargaining framework, both $(\mathbf{x}',\mathbf{y}')$ and $(\mathbf{x}^*,\mathbf{y}^*)$ are Pareto optimal.
 
 By definition, \textit{KS} fair solution indicates that
 \begin{align}
   & \min \{U'_i - U_i^{w}, ...\} > \min \{U^*_i - U_i^{w}, ...\}, \quad \forall v_i \in V \label{eq:pf:1}
 \end{align}
 By definition, egalitarian fair solution indicates that
 \begin{align}
   & \min \{ U^*_i - U_i^{w}, ...\} = U^*_i - U_i^{w} = U^*_j - U_j^{w}, \forall v_i, v_j \in V \label{eq:pf:2}
 \end{align}
 (\ref{eq:pf:1}), (\ref{eq:pf:2}) $\Longrightarrow$
 \begin{align}
   & U'_i - U_i^{w} \geq U^*_i - U_i^{w}, \quad \forall v_i \in V \label{eq:pf:5} \\
   & U'_i - U_i^{w} > U^*_i - U_i^{w}, \quad \exists v_i \in V \label{eq:pf:3}
 \end{align}
 Inequality (\ref{eq:pf:3}) contradicts with the fact that $(\mathbf{x}^*,\mathbf{y}^*)$ is Pareto optimal. So the assumption does not hold. $(\mathbf{x}^*,\mathbf{y}^*)$ must be both egalitarian fair and \textit{KS} fair. 
 
Because $(\mathbf{x}^*,\mathbf{y}^*)$ is already a Pareto optimal solution, both \textit{MF} and \textit{KS} are the same. Therefore, in a fair collaborative game $(\Omega,u^0)$, \textit{EF} is sufficient for \textit{MF}.
I.e., $EF \Rightarrow MF$.
 \end{proof}

 \subsection{Proof of Lemma \ref{thm:0}.}
 
 \begin{proof}
   The proof is trivial. Since $U_i$ in eq.~(\ref{eq:util}) is affine and positive, the non-negative weighted sum of $U_i$ is still affine and positive. All the affine functions are log-concave. So the objective function (\ref{eq:nash:max}) is concave.
 
   In addition, all (\ref{eq:cache})(\ref{eq:dst})(\ref{eq:fsb})(\ref{eq:int2}) and (\ref{eq:int1}) are defined over a set of compact and convex constraints. Therefore, problem (\ref{eq:nash:max}) is a convex optimisation problem.
 \end{proof}

 \subsection{Proof of Theorem \ref{eq:thm:kkt}.}
 
 \begin{proof}
 Obviously caching decision space $[0,1] \subset \mathbb{R}_+$ is a nonempty, compact and convex set. Since the objective function (\ref{eq:nash:max}) is a continuously differentiable concave function, and all the constraints on the variables are affine, Karush-Kuhn-Tucker (KKT) conditions are necessary and sufficient for the existence of an optimal solution.
 
 To derive the optimum of a function with constraints, we first derive the Lagrangian $\mathcal{L}(\cdot)$ of eq.~(\ref{eq:nash:max}). Let $\boldsymbol{\alpha} \succeq 0$, $\boldsymbol{\beta} \succeq 0$, $\boldsymbol{\gamma} \succeq 0$, $\boldsymbol{\delta} \succeq 0$ and $\boldsymbol{\lambda} \succeq 0$ be the KKT multipliers associated with constraints. Their subscripts are self-explained by the corresponding constraints associated with. Then we have
 \begin{align*}
   & \mathcal{L}(\mathbf{x}, \boldsymbol{\lambda}, \boldsymbol{\alpha}, \boldsymbol{\beta}, \boldsymbol{\gamma}, \boldsymbol{\delta}) = \\
   &\sum_{v_i \in V} \ln ( U_i - u_i^0 ) - \sum_{v_i \in V} \sum_{v_j \in N_i} \sum_{o_k \in O} \lambda_{i,j,k} (y_{i,j,k} - x_{j,k}) \\
   & - \sum_{v_i \in V} \alpha_{i} (\sum_{o_k \in O} x_{i,k} - C_i) - \sum_{v_i \in V} \sum_{o_k \in O} \beta_{i,k} (\sum_{v_j \in N_i} y_{i,j,k} - 1) \\
   & - \sum_{v_i \in V} \sum_{o_k \in O} \gamma_{i,k} (x_{i,k} - 1) + \sum_{v_i \in V} \sum_{v_j \in N_i} \sum_{o_k \in O} \delta_{i,j,k} y_{i,j,k}
 \end{align*}
 Note we dropped constraints $x_{i,j} \geq 0$ and $y_{i,j,k} \leq 1$ in making the Lagrangian because constraints (\ref{eq:dst}) and (\ref{eq:fsb}) make them redundant. In the following derivation, we let $\tau_{i,k} = U_i - u_i^0 - w_{i,k} x_{i,k}$ and $\tau_{i,k}' = U_i - u_i^0 - \frac{w_{i,k}}{l_{i,j}} y_{i,j,k}$ for the simplicity of representation. For the objective function to reach its optimum, first order necessary and sufficient conditions are
 \begin{align*}
   & \nabla \mathcal{L}(\mathbf{x}, \boldsymbol{\lambda}, \boldsymbol{\alpha}, \boldsymbol{\beta}, \boldsymbol{\gamma}, \boldsymbol{\delta}) = 0 \\
   & \Longleftrightarrow  \frac{\partial \mathcal{L}}{\partial x_{i,k}} = 0, \forall v_i, v_j \in V, \forall o_k \in O \\
   & \Longleftrightarrow \frac{w_{i,k}}{U_i - u_i^0} + \sum_{v_j \in N^+_i} \lambda_{j,i,k} - \alpha_i -\gamma_{i,k} = 0 \\
   & \Longleftrightarrow x^*_{i,k} = \frac{1}{\alpha_i + \gamma_{i,k} - \sum_{v_j \in N^+_i} \lambda_{j,i,k}} - \frac{\tau_{i,k}}{w_{i,k}}
 \end{align*}
 with complementary slackness
 \begin{align}
 \label{eq:kkt:sys}
 \begin{cases}
    \lambda_{i,j,k} (y_{i,j,k} - x_{j,k}) = 0, & \forall v_i, v_j \in V, \forall o_k \in O \\
    \alpha_{i} (\sum_{o_k \in O} x_{i,k} - C_i) = 0, & \forall v_i \in V, \forall o_k \in O \\
    \beta_{i,k} (\sum_{v_j \in N_i} y_{i,j,k} - 1) = 0, & \forall v_i \in V, \forall o_k \in O \\
    \gamma_{i,k} (x_{i,k} - 1) = 0, & \forall v_i \in V, \forall o_k \in O \\
    \delta_{i,j,k} y_{i,j,k} = 0, & \forall v_i, v_j \in V, \forall o_k \in O
 \end{cases}
 \end{align}
 
 Similarly, we can derive the optimal $y^*_{i,j,k}$ as $x^*_{i,k}$. The optimal caching strategy $(\mathbf{x^*},\mathbf{y^*})$ of the network can be derived by solving the equation system (\ref{eq:kkt:sys}) for all the nodes.
 \end{proof}

 \subsection{Proof of Theorem \ref{eq:thm:converge}.}
 
 \begin{proof}
 
 To prove convergence, we first prove the gradient of the dual function is bounded by a constant $K$, namely the dual function $d(\boldsymbol{\lambda})$ is K-Lipschitz continuous. Second, we show that given the diminishing step size, the Euclidean distance between the optimum $d(\boldsymbol{\lambda}^*)$ and the best value $d(\boldsymbol{\lambda}^{\circ})$ achieved in all previous iterations converges to zero in limit.
 
 Since the primal (\ref{eq:nash:9}) is strictly convex and all constraints are linear, dual $d(\boldsymbol{\lambda})$ is strictly concave and differentiable.
 \begin{align}
   & \frac {\partial d(\boldsymbol{\lambda})}{\partial \lambda_{i,j,k}} = y_{i,j,k} - x_{i,k} \Longrightarrow \left| \frac {\partial d(\boldsymbol{\lambda})}{\partial \lambda_{i,j,k}} \right| \leq 1
 \end{align}
 By Mean value theorem, there exists $\mathbf{c} \in (\boldsymbol{\lambda}, \boldsymbol{\lambda}')$ such that
 \begin{align}
   & d(\boldsymbol{\lambda}) - d(\boldsymbol{\lambda}') = \nabla d(\mathbf{c})^T (\boldsymbol{\lambda} - \boldsymbol{\lambda}')
 \end{align}
 By Cauchy--Schwarz inequality, let $n = |O| \times |V|^2$, we have
 \begin{align}
   \| d(\boldsymbol{\lambda}) - d(\boldsymbol{\lambda}') \|_2 & = \| \nabla d(\mathbf{c})^T (\boldsymbol{\lambda} - \boldsymbol{\lambda}') \|_2 \\
   & \leq \| \nabla d(\mathbf{c}) \|_2 \| \boldsymbol{\lambda} - \boldsymbol{\lambda}' \|_2 \\
   & \leq \sqrt{n} \| \boldsymbol{\lambda} - \boldsymbol{\lambda}' \|_2
 \end{align}
 $\| \cdot \|_2$ above denotes the Euclidean norm. Therefore, $d(\boldsymbol{\lambda})$ is K-Lipschitz continuous and Lipschitz constant $K = \sqrt{n}$. Let $\boldsymbol{\lambda}^*$ denote the maximiser of dual function $d(\boldsymbol{\lambda})$, then
 \begin{align}
   & \| \boldsymbol{\lambda}^{(t+1)} - \boldsymbol{\lambda}^* \|_2^2 = \| (\boldsymbol{\lambda}^{(t)} + \xi_k \mathbf{h}^{(t)})_+ - \boldsymbol{\lambda}^* \|_2^2 \\
   & \leq \| \boldsymbol{\lambda}^{(t)} + \xi_k \mathbf{h}^{(t)} - \boldsymbol{\lambda}^* \|_2^2 \label{eq:pff:prj} \\
   & = \| \boldsymbol{\lambda}^{(t)} - \boldsymbol{\lambda}^* \|_2^2 + 2 \xi_k \mathbf{h}^{(t)T} (\boldsymbol{\lambda}^{(t)} - \boldsymbol{\lambda}^*) + \xi_k^2 \| \mathbf{h}^{(t)} \|_2^2 \\
   & \leq \| \boldsymbol{\lambda}^{(t)} - \boldsymbol{\lambda}^* \|_2^2 + 2 \xi_k (d(\boldsymbol{\lambda}^{(t)}) - d(\boldsymbol{\lambda}^*)) + \xi_k^2 \| \mathbf{h}^{(t)} \|_2^2 \label{eq:pff:23}
 \end{align}
 Inequality (\ref{eq:pff:prj}) comes from the fact that projection of a point onto $\mathbb{R}_+^{|O||V|^2}$ makes it closer to the optimal point in $\mathbb{R}_+^{|O||V|^2}$. Apply inequality~(\ref{eq:pff:23}) recursively, we have
 \begin{align*}
   & \| \boldsymbol{\lambda}^{(t+1)} - \boldsymbol{\lambda}^* \|_2^2 \leq \\
   & \| \boldsymbol{\lambda}^{(1)} - \boldsymbol{\lambda}^* \|_2^2 + 2 \sum_{i=1}^{k} \xi_i (d(\boldsymbol{\lambda}^{(i)}) - d(\boldsymbol{\lambda}^*)) + \sum_{i=1}^{k} \xi_i^2 \| \mathbf{h}^{(i)} \|_2^2
 \end{align*}
 Because $ \| \boldsymbol{\lambda}^{(t+1)} - \boldsymbol{\lambda}^* \|_2^2 \geq 0$ and $\sum_{i=1}^{k} \xi_i > 0$, and let $d(\boldsymbol{\lambda}^{\circ}) = \max_{0 \leq i < k} {d(\boldsymbol{\lambda}^{(i)})}$, then
 \begin{align*}
   & 2 \sum_{i=1}^{k} \xi_i (d(\boldsymbol{\lambda}^*) - d(\boldsymbol{\lambda}^{\circ})) \leq \| \boldsymbol{\lambda}^{(1)} - \boldsymbol{\lambda}^* \|_2^2 + \sum_{i=1}^{k} \xi_i^2 \| \mathbf{h}^{(i)} \|_2^2 \\
   & \quad \quad \Longrightarrow d(\boldsymbol{\lambda}^*) - d(\boldsymbol{\lambda}^{\circ}) \leq \frac{\| \boldsymbol{\lambda}^{(1)} - \boldsymbol{\lambda}^* \|_2^2 + \sum_{i=1}^{k} \xi_i^2 \| \mathbf{h}^{(i)} \|_2^2}{2 \sum_{i=1}^{k} \xi_i} \\
   & \quad \quad \Longrightarrow d(\boldsymbol{\lambda}^*) - d(\boldsymbol{\lambda}^{\circ}) \leq \frac{\| \boldsymbol{\lambda}^{(1)} - \boldsymbol{\lambda}^* \|_2^2 + K^2 \sum_{i=1}^{k} \xi_i^2}{2 \sum_{i=1}^{k} \xi_i}
 \end{align*}
 $d(\boldsymbol{\lambda}^*) - d(\boldsymbol{\lambda}^{\circ}) \rightarrow 0$ if we choose a diminishing step size which lets $\xi_i \rightarrow 0$ and $\sum_{1}^{\infty} \xi_i = \infty$, then $\frac{\sum_{1}^{\infty} \xi^2_i}{\sum_{1}^{\infty} \xi_i} = 0$. (e.g. we can let $\xi_{i} = \frac{\xi_0}{i}$, then $\sum_{1}^{\infty} \xi_i = \infty$ and $\sum_{1}^{\infty} \xi_i^2 = \frac{\pi^2}{6}$.) Since the duality gap is zero, eventually the primal problem will converge to its optimum when its dual problem converges.
 \end{proof}


%

\begin{IEEEbiography}
[{\includegraphics[width=1in,height=1.25in,clip,keepaspectratio]{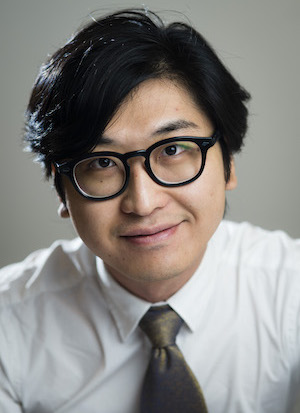}}]
{Liang Wang}
is a research associate in the Computer Laboratory at University of Cambridge, United Kingdom. In 2003, he received his BEng in Computer Science and Mathematics from Tongji University, Shanghai, China. Later, he received both his MSc and PhD degrees in Computer Science from University of Helsinki, Finland in 2011 and 2015 respectively. Liang's research interests include system and network optimisation, modelling and analysis of complex networks, information-centric networks, and distributed data processing, and so on. 
\end{IEEEbiography}

\begin{IEEEbiography}
[{\includegraphics[width=1in,height=1.25in,clip,keepaspectratio]{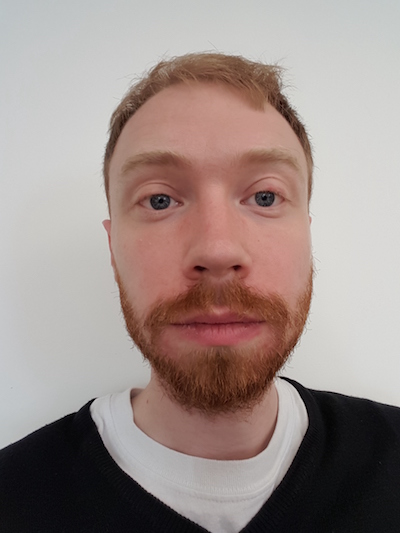}}]
{Gareth Tyson}
is a Lecturer at Queen Mary University of London. He receieved his PhD from Lancaster University, UK in 2010. His research centres on system measurements and design, looking at topics ranging from network operations to social media. He serves as a reviewer and program committee member for a number of prominent journals such as IEEE/ACM ToN, IEEE JSAC, IEEE TPDS, IEEE TNSM ACM TMM, IEEE TC. He recieved the Outstanding Reviewer Award at ICWSM'16.
\end{IEEEbiography}

\begin{IEEEbiography}
[{\includegraphics[width=1in,height=1.25in,clip,keepaspectratio]{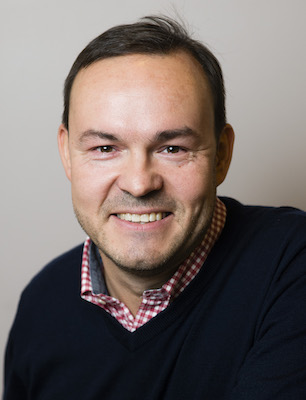}}]
{Jussi Kangasharju}
received his MSc from Helsinki University of Technology in 1998. He received his Diplome d'Etudes Approfondies (DEA) from the Ecole Superieure des Sciences Informatiques (ESSI) in Sophia Antipolis in 1998. In 2002 he received his PhD from University of Nice Sophia Antipolis/Institut Eurecom. In 2002 he joined Darmstadt University of Technology (TUD), first as post-doctoral researcher, and from 2004 onwards as assistant professor. Since June 2007 Jussi is a professor of computer science at University of Helsinki. 
\end{IEEEbiography}

\begin{IEEEbiography}
[{\includegraphics[width=1in,height=1.25in,clip,keepaspectratio]{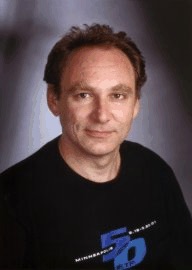}}]
{Jon Crowcroft}
(SM’95–F’04) received the B.S. degree in physics from  University of Cambridge in 1979, and the MSc degree in computing and PhD degree from UCL in 1981 and 1993, respectively.
He has been a Professor at the University of Cambridge since 2001. He has worked in the area of Internet support for multimedia communications for over 30 years. Prof.\ Crowcroft is a Fellow of the Royal Society, the Association for Computing Machinery (ACM), the British Computer Society, the IET, and the Royal Academy of Engineering.
\end{IEEEbiography}

\end{document}